\documentclass[11pt, a4paper]{article}
\usepackage{amsthm,amsfonts,amsmath,amssymb,bookmark,appendix}
\usepackage{verbatim}

\theoremstyle{plain}
\newtheorem{theorem}{Theorem}
\newtheorem{lemma}[theorem]{Lemma}
\newtheorem{assumption}{Assumption}

\newtheorem{proposition}[theorem]{Proposition}

\theoremstyle{remark}
\newtheorem*{remark}{Remark}

\DeclareMathOperator{\supp}{supp}

\DeclareMathOperator{\loc}{loc}

\newcommand{\vphi}{\varphi}
\newcommand{\norm}[1]{\lVert #1 \rVert}
\newcommand{\Norm}[1]{\left\lVert #1 \right\rVert}

\newcommand{\eps}{\varepsilon}
\newcommand{\R}{\mathbb{R}}
\newcommand{\C}{\mathbb{C}}

\newcommand{\ue}{\underline{e}_K}
\newcommand{\bx}{\boldsymbol{x}}
\newcommand{\by}{\boldsymbol{y}}
\newcommand{\bz}{\boldsymbol{z}}
\newcommand{\br}{\boldsymbol{r}}
\newcommand{\sx}{\mathrm{x}}
\newcommand{\cE}{\mathcal{E}}

\newcommand{\LL}{\textup{LL}}
\newcommand{\hs}{\textup{hs}}
\newcommand{\hd}{\textup{hd}}
\newcommand{\calpha}{\tau}

\begin{document}

\title{Lieb-Thirring Bounds for Interacting Bose Gases\footnote{D.L. is supported by the grant KAW 2010.0063 from the Knut and Alice Wallenberg Foundation and the Swedish Research Council grant no. 2013-4734. F.P. was supported by the Swedish Research Council grant no. 2012-3864. J.P.S. is supported by the ERC AdvGrant project no. 321029.}}
\date{}

\author{D. Lundholm\footnote{dogge@math.kth.se, KTH Royal Institute of Technology, Sweden}, F. Portmann\footnote{fabianpo@math.kth.se, KTH Royal Institute of Technology, Sweden}, J. P. Solovej\footnote{solovej@math.ku.dk, University of Copenhagen, Denmark}}

\maketitle

\begin{abstract}
	We study interacting Bose gases and prove lower bounds for the kinetic plus interaction 
	energy of a many-body wave function in terms of its particle density. 
	These general estimates are then applied to 
	various types of interactions, including hard sphere (in 3D) and hard disk (in 2D) 
	as well as a general class of homogeneous potentials.
\end{abstract}

\tableofcontents

\section{Introduction}
The last two decades have seen an explosion of results on the 
experimental physics of ultracold atomic gases and on
the theoretical and mathematical physics of many-body 
quantum statistical mechanics. 
The concrete realization in 1995 \cite{Anderson-et-al:95,Davies-et-al:95} of 
Bose-Einstein condensation in trapped dilute
gases offered a set-up which is quite different from the 
standard textbook treatment of the ideal
Bose gas, originating in work of Bose and Einstein from 1924-25, 
and brought a demand for a better
understanding of the effects of interactions in condensates
beyond the pioneering work of Bogoliubov from 1947
and of e.g. Gross and Pitaevskii from 1961.
Also various more extreme conditions and geometries 
have been experimentally realized in recent years,
such as effectively one- and two-dimensional systems, 
which previously were considered to be only of purely theoretical interest.
We refer to the reviews \cite{DGPS:99,Bloch-Dalibard-Zwerger:08} and the 
book \cite{Pitaevskii-Stringari:03} for comprehensive introductions and 
historical surveys on the physical 
aspects of this vast topic.

The quantum Hamiltonian for $N$ identical bosons in an external one-body trapping potential 
$V$ in $d$ dimensions, interacting with symmetric pair potential $W$, 
is given by
\begin{align} \label{eq:Hamiltonian}
	\hat{H} = \hat{T} + \hat{V} + \hat{W} 
	= \sum_{j=1}^N\left( -\mu\Delta_j + V(\bx_j)\right) + \sum_{j<k} W(\bx_j - \bx_k),
\end{align}
with $\mu := \hbar^2/(2m)$ the reduced mass.
It is acting on $\psi \in L^2_{\textup{sym}}(\R^{dN})$,
i.e. wave functions which are square-integrable and totally symmetric with
respect to all particle labels,
in accordance with bosonic statistics\footnote{However, 
we will frequently use the well-known fact that the ground state of 
\eqref{eq:Hamiltonian} with and without this symmetrization requirement is 
the same \cite{Lieb-Seiringer-Solovej-Yngvason:05}.}.
We denote the one-body density of the state $\psi$ 
(henceforth always assumed to be normalized) by 
\begin{align*}
	\rho(\bx) 
	:= \sum_{j=1}^N \int_{\R^{d(N-1)}} |\psi(\bx_1,\ldots,\bx_{j-1},\bx,\bx_{j+1},\ldots,\bx_N)|^2 
	\prod_{k \neq j} d\bx_k
\end{align*}
and the mean density by $\bar{\rho} := \frac{1}{N}\int_{\R^d} \rho^2$.
Assuming that the interaction potential $W$ is sufficiently rapidly decaying,
it can on large scales and low energies be characterized by its scattering length $a$ 
(see the Appendix for a definition), where $a>0$ for repulsive interactions.
The first mathematically rigorous and reasonably sharp bounds for the 
ground state energy of an interacting Bose gas were derived 
in 1957 by Dyson \cite{Dyson:57} for the case of hard sphere potentials 
(where $a$ is then the range, i.e. the diameter of the sphere),
and where a homogeneous gas ($V=0$) was considered
in the thermodynamic limit.
His result has later been sharpened
as well as extended to other potentials
and dimensionalities in the dilute gas limit 
\cite{Lieb-Liniger:63,Lieb-Yngvason:98,Lieb-Yngvason:00,Lieb-Yngvason:01,Lieb-Seiringer-Yngvason:03,Lieb-Seiringer-Yngvason:04,Schnee-Yngvason:06}. 
These mathematical works are summarized in \cite{Lieb-Seiringer-Solovej-Yngvason:05}.

Although the full range of physical dimensionalities 
$d\in\{1,2,3\}$ is relevant for experiments, 
we choose to simplify the present discussion
by focusing on the $d=3$ (3D) case.
In the dilute limit $a^3\bar{\rho} \to 0$ with $N \to \infty$ 
and $Na/L \sim const.$, $L$ denoting the length scale of $V$,
the ground state of the trapped interacting Bose gas is correctly described
by the Gross-Pitaevskii energy functional 
\begin{align} \label{eq:GP-functional}
	\cE_{\textup{GP}}[\phi] 
	:= \int_{\R^3} \left(\mu|\nabla \phi|^2 + V |\phi|^2 + 4\pi \mu a|\phi|^4 
	\right),
\end{align}
acting on functions $\phi: \R^3 \to \C$
constrained by $\int_{\R^3} |\phi|^2 = N$.
The last term corresponds to the scattering length approximation to the
interaction energy.
This description is asymptotically correct in the sense that 
the true ground state $\psi$ of \eqref{eq:Hamiltonian} has energy 
$E_0 \approx \cE_{\textup{GP}}[\phi_{\textup{GP}}]$
and density $\rho(\bx) \approx |\phi_{\textup{GP}}(\bx)|^2$,
where $\phi_{\textup{GP}}$ 
is the unique minimizer of the functional \eqref{eq:GP-functional}.
Moreover, if the scattering length is comparatively large,
$Na/L \to \infty$ while $a^3\bar{\rho} \to 0$, then the 
gradient term in \eqref{eq:GP-functional} becomes negligible and the
so-called Thomas-Fermi approximation becomes valid, in which the ground state
density of the system is correctly described as the minimizer of the functional 
\begin{align} \label{eq:TF-functional}
	\cE_{\textup{TF}}[\rho] := \int_{\R^3} \left( V\rho + 4\pi\mu a\rho^2 \right),
\end{align}
subject to the constraints $\rho \ge 0$ and $\int_{\R^3} \rho = N$.
We will in the following assume that $\mu=1$, 
which can be achieved by an appropriate scaling of the energy.

These functionals provide a useful and remarkably precise description of 
ground state properties for experiments in the 
zero-temperature low-density regime \cite{DGPS:99},
although their mathematical validity depends crucially on the smallness of 
$a^3\bar\rho$ \cite{Lieb-Seiringer-Solovej-Yngvason:05}.
Corrections to the Gross-Pitaevskii term in \eqref{eq:GP-functional} 
and \eqref{eq:TF-functional} are available as perturbative expansions 
\cite{Huang-Yang:57, Lee-Huang-Yang:57, Brueckner-Sawada:57,
Brueckner-Sawada-2:57, Beliaev:58,Wu:59, Hugenholty-Pines:59, Girardeau-Arnowitt:59,
Lee-Yang:60, Giuliani-Seiringer:09, Yau-Yin:09}, 
although there is no control on the convergence of such expansions.
For fermionic systems, functionals expressing the energy in terms of the 
density alone are widely used in quantum chemistry, where such density 
functional theories have been
very successful in describing both the equilibrium states and the
dynamics of atomic and molecular systems.
Our objective in this work is to prove rigorous lower bounds for the energy of
interacting Bose gases described by \eqref{eq:Hamiltonian} and a given state $\psi$, 
in terms of (more or less explicit) energy functionals of the density of $\psi$,
and with general validity irrespective of the density of the gas
(although these explicit bounds may become comparatively weak in the high-density limit).
This will be achieved using a local approach to the interaction energy.
To gain some insight from
a much better understood situation, 
consider the physically very
different case of $N$ non-interacting fermions,
for which the Pauli principle together with the uncertainty principle
conspire to yield
a strong lower bound for the kinetic energy 
--- the celebrated Lieb-Thirring inequality \cite{Lieb-Thirring:75,Lieb-Thirring:76},
\begin{align} \label{eq:classic-LT}
	\langle \psi, \hat{T} \psi \rangle 
	:= \sum_{j=1}^N \int_{\R^{dN}}\left|\nabla_j \psi\right|^2\,d\sx
	\,\ge\, C_{\textup{LT}} \int_{\R^d} \rho(\bx)^{1+2/d}\,d\bx,
\end{align}
for antisymmetric $\psi$ and a constant $C_{\textup{LT}} > 0$ independent of $N \ge 1$.
The (fermionic) Thomas-Fermi approximation corresponds to the
right hand side of \eqref{eq:classic-LT},
with a semi-classical constant $C_{\textup{TF}} \ge C_{\textup{LT}}$.
Bounds of the form \eqref{eq:classic-LT} are extremely useful when further
interactions are involved, such as in the rigorous proof of stability of matter
with Coulomb interactions 
(see, e.g., \cite{Lieb-Seiringer:10} for an updated review)
which was actually the original motivation for the Lieb-Thirring inequality 
in \cite{Lieb-Thirring:75}. In contrast to this simplified approach,
the original proof of stability of matter, due to Dyson and Lenard in 1967 
\cite{Dyson-Lenard:67,Dyson-Lenard:68} 
(see also \cite{Dyson:68,Lenard:73}),
relied on a purely local consequence of the exclusion principle, 
based on the following simple bound for 
an $n$-fermion wave function on a local domain $Q \subset \R^d$
(with $c_Q>0$ depending only on the shape of $Q$):
\begin{align}\label{eq:dl_fermion_bound_domain}
	\int_{Q^n}\sum_{j=1}^n|\nabla_j\psi|^2 d\sx 
	\geq (n-1)\frac{c_Q}{|Q|^{2/d}}\int_{Q^n}|\psi|^2\,d\sx.
\end{align}
This local exclusion principle
was recently adopted to other types of 
identical particles for which the quantum statistical properties 
can be modeled by local interactions between bosons,
such as anyons in two dimensions \cite{Lundholm-Solovej:anyon}
and particles in one dimension exhibiting intermediate statistics 
\cite{Lundholm-Solovej:extended}
(see also \cite{Lundholm-Solovej:exclusion} for a physical review),
and it was in this work recognized that the original 
local approach due to Dyson and Lenard
is actually sufficient for proving the Lieb-Thirring inequality \eqref{eq:classic-LT},
albeit with a weaker constant than that in \cite{Lieb-Thirring:75}.

Our approach in this paper is to extend the methods in 
\cite{Lundholm-Solovej:anyon, Lundholm-Solovej:extended} 
to more general interacting Bose gases
by considering a scale-normalized
two-particle interaction energy on a local cube $Q$, $e_2(|Q|;W)$,
defined in Section~\ref{sec:basic_energy_est} as $|Q|^{2/d}$ 
times the Neumann ground state energy for two bosons
on $Q$ with the repulsive\footnote{In all our applications 
we have $W \ge 0$ and radially symmetric, 
however the main result holds under the weaker conditions 
given in Assumptions \ref{ass:local_exclusion_alpha} 
and \ref{ass:local_uncert_alpha}.}
pair interaction $W \ge 0$
(which hence for $n=2$ replaces the pure Pauli repulsion on the l.h.s. of 
\eqref{eq:dl_fermion_bound_domain}).
Furthermore, despite the general difficulty in computing $e_2$ analytically,
it turns out to be sufficient for our purposes to have a lower bound of the form
\begin{align} \label{e2-bound}
	e_2(|Q|;W) \ge e(\gamma(|Q|)), 
	\quad \text{with} \quad
	\gamma(|Q|) = \calpha |Q|^{(2-\alpha)/d},
\end{align}
in terms of a simpler function $e(\gamma)$ depending only on a dimensionless 
parameter $\gamma$ for some dimensionful constant $\calpha > 0$ 
and scaling parameter $\alpha > 0$.
\textbf{Our main result}, given in Section~\ref{sec:general_LT_bounds}, is the derivation of a 
general family of Lieb-Thirring type energy bounds of the form
\begin{align} \label{eq:main_result}
	\langle\psi,(\hat{T}+\hat{W})\psi\rangle 
	\geq C \int_{\R^d} e(\gamma(2/\rho(\bx))) \rho(\bx)^{1+2/d} \,d\bx,
\end{align}
where $C>0$ depends only on $d$, $\alpha$, and on an upper bound for $e$ 
(see Theorem~\ref{thm:L-T_global_alpha}).
These general energy inequalities are then in Section~\ref{sec:applications}
applied to concrete examples where explicit lower bounds 
of the appropriate form \eqref{e2-bound}
can be computed, such as the hard sphere (3D) and hard disk (2D) potentials,
the Lieb-Liniger model for point interacting bosons (1D),
as well as a family of homogeneous potentials, $W(\bx) \propto |\bx|^{-\beta}$, 
in 3D as well as 2D. For instance, we prove (see Theorem~\ref{thm:L-T_HS})
that for the hard-sphere interaction with range $a$,
\begin{align*}
	\langle\psi,(\hat{T}+\hat{W})\psi\rangle 
	\geq C \int_{\R^3} \min\left\{ \frac{2^{2/3}}{\sqrt{3}} 
	a\rho(\bx)^{2}, \pi^2 \rho(\bx)^{5/3} \right\} d\bx,
\end{align*}
while for the hard-disk (see Theorem~\ref{thm:L-T_hd}) we have
\begin{align*}
	\langle\psi,(\hat{T}+\hat{W})\psi\rangle 
	\geq C \int_{\R^2} \frac{ 2\rho(\bx)^{2} }
	{ 2 + \left( -\ln(a\rho(\bx)^{1/2}/2) \right)_+ }\,d\bx.
\end{align*}
In Section~\ref{sec:counter_ex} we consider the sharpness of the forms of these bounds
by means of counterexamples.
Some auxiliary results concerning uncertainty principles and 
scattering lengths are given in the Appendix.

Let us finally comment on the physical interest of the bounds presented here
from the perspective of applications.
First, the extension of Lieb-Thirring type inequalities to the bosonic
context opens up for the application of useful techniques 
originally developed for fermionic systems.
For example, as shown in Section \ref{sec:inverse-square}, 
an inverse-square repulsive 
interaction for bosons yields a bound analogous to 
the fermionic kinetic energy inequality 
(cf. the r.h.s. of \eqref{eq:classic-LT}),
$$
	\langle\psi, (\hat{T} + \hat{W})\psi \rangle 
	\ge C \int_{\R^d} \rho(\bx)^{1+2/d} \,d\bx,
$$
and therefore 
also a corresponding
Lieb-Thirring inequality for the negative part of an external 
one-body potential $V$,
$$
	\hat{T} + \hat{W} + \hat{V} 
	\ge -C' \int_{\R^d} |V_-|^{1+d/2} \,d\bx.
$$
Thus, following the conventional approach for fermions 
(see e.g. \cite{Lieb-Seiringer:10}), 
these bounds can also be applied with additional Coulomb 
interactions to prove thermodynamic stability for a system of charged bosons
with inverse-square repulsive cores.

Second, although we have here mostly focused on the role of the 
interaction potential $W$,
we can also consider 
more general 
applications of these bounds 
in the presence of external potentials, 
for example describing experiments with gases trapped by a confining potential $V$.
Given a bound of the form \eqref{eq:main_result} 
we have for an arbitrary state $\psi$:
\begin{align} \label{eq:LT_w_external_pot}
	\langle\psi, (\hat{T} + \hat{W} + \hat{V})\psi \rangle 
	\ge \int_{\R^d} \left\lbrack C e(\gamma(2/\rho(\bx)))\rho(\bx)^{1+2/d} + V(\bx)\rho(\bx) \right\rbrack d\bx,
\end{align}
and can hence obtain lower bounds for the ground state energy of the system
by minimizing the r.h.s. subject to the constraint $\int_{\R^d} \rho = N$.
This can at the same time provide useful estimates for the ground state density $\rho$.
As a concrete example, we can consider the Lieb-Liniger model 
(see Section \ref{sec:Lieb-Liniger})
for which \eqref{eq:LT_w_external_pot} becomes
\begin{align} \label{eq:LL_w_external_pot}
	\langle\psi, (\hat{T} + \hat{W} + \hat{V})\psi \rangle 
	\ge \int_{\R} \left\lbrack C_{\LL} \xi_{\LL}(2\eta/\rho(x))^2 \rho(x)^3 + V(x)\rho(x) \right\rbrack dx,
\end{align}
an explicit convex functional of $\rho$ which is tractable for minimization 
given an external potential $V$.
In contrast, the Lieb-Liniger model has, despite its relative simplicity,
only been exactly solved for the 
ground state energy and density in the absence of any external potential.
Unfortunately, the universal constants 
appearing in these bounds are far
from the optimal ones, and therefore the resulting explicit bounds might be 
far lower than the exact ground state energies
(although observe that even the optimal constant $C_{\textup{LT}}$ 
in the original kinetic energy inequality \eqref{eq:classic-LT} for $d=3$
has yet to be shown to be as large as its conjectured value $C_{\textup{TF}}$). 
However, the present results may be seen
as a first step in this direction.

We also remark that the techniques employed here can be extended to the case of 
fractional kinetic energy operators,
for instance to prove a Lieb-Thirring inequality for relativistic bosons
with Coulomb repulsion,
$$
	\left\langle\psi, \left( \sum_{j=1}^N \sqrt{-\Delta_j} 
		+ \sum_{j<k}\frac{1}{|\bx_j-\bx_k|} \right)\psi \right\rangle 
	\ge C \int_{\R^d} \rho(\bx)^{1 + 1/d} \,d\bx,
$$
for $d \ge 1$; see \cite{Lundholm-Nam-Portmann}.
We would like to thank P. T. Nam for discussions on this point.

\paragraph{Aknowledgements}
We are grateful to R. Seiringer and A. Trombettoni for comments and discussions. 
We also thank the organizers of the IHP trimester program 
``Variational and Spectral Methods in Quantum Mechanics'' for providing a stimulating 
atmosphere.

\section{Basic Energy Estimates}\label{sec:basic_energy_est}
In this section we establish a basic energy estimate in terms of the two-particle 
energy, which is valid for all nonnegative interaction potentials $W$. 
We refer to this as a local exclusion principle, 
in analogy to the discussion in the introduction. 
Namely, the repulsive interaction prevents 
two bosons from occupying the same local (zero-energy) state.

Let $Q \subset \R^d$ denote a cube of side length $|Q|^{1/ d}$. We define 
the (scale-normalized) two-particle energy on $Q$ 
by\footnote{As usual, we denote by $H^1(\Omega)$ the Sobolev space 
of square-integrable functions on $\Omega \subseteq \R^n$ with 
square-integrable first derivatives, 
i.e. the form domain of the 
Laplacian on $\Omega$ with Neumann boundary conditions. 
We will, unless stated otherwise, define all considered operators via their 
natural quadratic forms.}
$$
	e_2(|Q|;W):=
	|Q|^{2/d}\!\!\! \inf_{\substack{\psi \in H^1(Q^2)\\ \int_{Q^2}|\psi|^2=1}}\,\int_{Q^2}\!\!\left[|\nabla_1\psi|
	^2 + |\nabla_2\psi|^2 + W(\bx_1-\bx_2)|\psi|^2\right]d\bx_1d\bx_2.
$$
Note that $e_2(|Q|;\lambda W)$ is monotone and concave in the parameter $\lambda$ and
$e_2(|Q|;0) = 0$.

\begin{lemma}\label{lem:estimate_by_e_2}
	Given any pair-interaction potential $W$ 
	and particle number $n \geq 2$, we have for any cube $Q$ 
	and all $\psi \in H^1(Q^n)$
	the estimate
	\begin{multline*}
		\int_{Q^n}\left[\sum_{j=1}^n|\nabla_j\psi|^2+\sum_{1\leq j<k\leq n}
		W(\bx_j-\bx_k)|\psi|^2\right]d\sx\\ 
		\geq \frac{n}{2}|Q|^{-2/d}e_2(|Q|;(n-1)W)\int_{Q^n}|\psi|^2\,d\sx.
	\end{multline*}
\end{lemma}
\begin{proof}
	Using the identity
	\begin{align*}
		(n-1) \sum_{j=1}^n |\nabla_j \psi|^2 
		= \sum_{1\leq j < k \leq n} (|\nabla_j \psi|^2 + |\nabla_k \psi|^2),
	\end{align*}
	we write
	\begin{multline*}
		\int_{Q^n}\left[\sum_{j=1}^n|\nabla_j\psi|^2+\sum_{1\leq j<k\leq n}W(\bx_j-\bx_k)|\psi|^2\right]d\sx\\
		= \int_{Q^n} \sum_{j<k}\left[\frac{1}{n-1}(|\nabla_j \psi|^2 + |\nabla_k \psi|^2) 
		+ W(\bx_j-\bx_k)|\psi|^2\right]\,d\sx\\
		= \frac{1}{(n-1)}\sum_{j<k}\int_{Q^{n-2}}\Bigg[\int_{Q^2}(|\nabla_j \psi|^2 + |\nabla_k \psi|^2\\
		+ (n-1)W(\bx_j-\bx_k)|\psi|^2)\,d\bx_jd\bx_k\Bigg]d\hat{\sx}\\
		\geq \frac{n}{2}|Q|^{-2/d}e_2(|Q|;(n-1)W)\int_{Q^n}|\psi|^2\,d\sx,
	\end{multline*}
	where $\hat{\sx} = (\bx_1,\ldots,\bx\!\!\!\!\diagup\!_j,\ldots,\bx\!\!\!\!\diagup\!_k,\ldots,\bx_N)$.
\end{proof}

Given any normalized $N$-particle wave function $\psi$, we define the local kinetic 
energy on $Q \subset \R^d$ as
\begin{align*}
	T_{\psi}^{Q} := \sum_{j=1}^N\int_{\R^{dN}}|\nabla_j\psi|^2\chi_{Q}(\bx_j)\,d\sx,
\end{align*}
where $\chi_Q$ denotes the characteristic function of the domain $Q$.
The local interaction energy on $Q$ is given by
\begin{align*}
	W_{\psi}^Q := \frac{1}{2}\sum_{j=1}^N
	\sum_{(j\neq)k=1}^N\int_{\R^{dN}}W(\bx_j-\bx_k)|\psi|^2\chi_Q(\bx_j)\,d\sx.
\end{align*}

\begin{theorem}[Local Exclusion]\label{thm:local_ex_TW}
	Let $W \geq 0$ and $N\geq 1$. Then for any finite cube $Q$ and all normalized
	$\psi \in H^1(\R^{dN})$ the local energy 
	$(T+W)_{\psi}^{Q} = T_{\psi}^Q + W_{\psi}^Q$ satisfies
	\begin{align}\label{eq:local_ex_TW}
		(T+W)_{\psi}^{Q} 
		\geq \frac{1}{2}\frac{e_2(|Q|;W)}{|Q|^{2/d}}\left(\int_{Q}\rho(\bx)\,d\bx - 1\right)_{+},
	\end{align}
	where $\rho$ is the density associated to $\psi$.
\end{theorem}
\begin{proof}
	We insert the partition of unity (cf. \cite{Dyson-Lenard:67,Lundholm-Solovej:anyon})
	\begin{align*}
		1 = \sum_{A \subseteq \{1,\dots, N\}} \prod_{l \in A}\chi_Q(\bx_l) 
		\prod_{l \notin A}\chi_{Q^c}(\bx_l)
	\end{align*}
	into the definition of $(T+W)_{\psi}^{Q}$ to obtain, using Lemma~\ref{lem:estimate_by_e_2},
	\begin{multline*}
		(T+W)_{\psi}^{Q} 
		= \sum_{A}\int_{\R^{dN}}\sum_{j \in A}\left[|\nabla_j\psi|^2 + \frac{1}{2}\sum_{(j
		\neq)k=1}^{N}W(\bx_j-\bx_k)|\psi|^2\right]\times\\
		\times\prod_{l \in A}\chi_Q(\bx_l)\prod_{l \notin A}\chi_{Q^c}(\bx_l)\,d\sx\\
		\geq \sum_{A} \int_{(Q^c)^{N-|A|}}\int_{Q^{|A|}}\!\!
		\left[\sum_{j \in A}|\nabla_j\psi|^2+ \frac{1}{2}\sum_{\substack{j,k\in A\\ j \neq k}}
		W(\bx_j-\bx_k)|\psi|^2\right]\prod_{l \in A}d\bx_l\prod_{l \notin A}d\bx_l\\
		\geq \sum_{A} \frac{|A|}{2}|Q|^{-2/d}e_2(|Q|;(|A|-1)W)\int_{(Q^c)^{N-|A|}}\int_{Q^{|A|}}|\psi|^2\,
		\prod_{l \in A}d\bx_l\prod_{l \notin A}d\bx_l.
	\end{multline*}
	The monotonicity of $e_2(|Q|;W)$ in the potential $W$ gives 
	\begin{align*}
		e_2(|Q|;(|A|-1)W) \geq e_2(|Q|;W), \quad \text{for $|A| \geq 2$}, 
	\end{align*}
	so that
	\begin{multline*}
		(T+W)_{\psi}^{Q} 
		\geq e_2(|Q|;W)|Q|^{-2/d}\sum_{|A|\geq 2} \frac{|A|}{2}\int_{(Q^c)^{N-|A|}}\int_{Q^{|A|}}|
		\psi|^2\,\prod_{l \in A}d\bx_l\prod_{l \notin A}d\bx_l\\
		\geq \frac{1}{2}e_2(|Q|;W)|Q|^{-2/d}
		\int_{\R^{dN}}\left(\sum_{j=1}^N\chi_Q(\bx_j)-1\right)|\psi|^2\,d\sx,
	\end{multline*}
	where we used $|A|/2 \geq (|A|-1)_+/2 \geq (|A|-1)/2$ in the sum over $A$, 
	and finally the partition of unity again. This proves \eqref{eq:local_ex_TW}.
\end{proof}
\begin{remark}
	Note that if $e_2(|Q|;\lambda W)$ turns out to be linear (or superlinear) in $\lambda$, then 
	the factor $1/2$ in \eqref{eq:local_ex_TW} can be removed, using $|A|(|A|-1)/2 \ge (|A|-1)_+$.
	This also applies for the lower bounds for $e_2$ 
	employed below.
\end{remark}
In addition to the above remark, it could be useful to point out precisely which 
sacrifices in energy have been made in obtaining the lower 
bound \eqref{eq:local_ex_TW}.
\begin{enumerate}
	\item In the first step of the proof, any interactions between particles inside $Q$ 
	and those outside $Q$ have been ignored (here the assumption $W \ge 0$ enters crucially). 
	This is expected to be a good approximation when the range of $W$ is small compared 
	to the size of $Q$, and hence in particular in the dilute limit.
	
	\item The lowest-energy contribution from the wave function to the two-particle energy 
	has been estimated, with Neumann boundary conditions, using 
	Lemma~\ref{lem:estimate_by_e_2}. 
	In the dilute limit (where such b.c. impose only a small error) this 
	corresponds to s-wave scattering.
	
	\item Higher-$n$-particle contributions in the wave function have been dominated by the 
	two-particle contribution $e_2$. By the above remark, this estimate is improved if one has  
	knowledge of the precise scaling behavior of $e_2$. We also note that the resulting 
	bound proves to be sufficient for many purposes, cf. the discussion for 
	fermions in \cite[Section~13]{Dyson:68}.
\end{enumerate}

\section{General Lieb-Thirring Bounds}\label{sec:general_LT_bounds}
It is in general difficult to obtain an explicit expression for $e_2(|Q|;W)$ as a 
function of $|Q|$, but it turns out that one can often bound the two-particle 
energy from below by a simpler function $e(\gamma)$, which only depends on a 
dimensionless parameter $\gamma$ and inherits monotonicity and concavity. 
For $\alpha >0$, we define
\begin{align}\label{def:gamma}
	\gamma(|Q|) := \calpha|Q|^{(2-\alpha)/d},
\end{align}
where $\calpha>0$ is an arbitrary constant. We hence want to obtain an estimate
\begin{align}\label{eq:two_particle_bound}
	e_2(|Q|;W) \geq e(\gamma(|Q|))
\end{align}
for some $\alpha > 0$ and $\calpha>0$, where $e(\gamma)$ is concave and 
monotone in $\gamma$, and we will show in this section that 
such an estimate is sufficient for 
deducing a Lieb-Thirring type bound for the energy in terms of the 
density.

For the remainder of this section we therefore make the following assumptions
on the potential $W$:
\begin{assumption}[Local Exclusion]\label{ass:local_exclusion_alpha}
	Given $W$, there exists a function $e(\gamma)$ with $\gamma$ as in
	\eqref{def:gamma} (with $\alpha, \calpha>0$), 
	where $e(\gamma)$ is monotone increasing and concave
	in $\gamma$ with $e(0)=0$, such that
	for any finite cube $Q$, any $N \geq 1$ and all normalized $\psi \in H^1(\R^{dN})$
	the local energy satisfies
	\begin{align*}
		(T+W)_{\psi}^{Q} 
		\geq \frac{1}{2} \frac{e(\gamma(|Q|))}{|Q|^{2/d}} \left(\int_Q \rho \ -1\right)_+,
	\end{align*}
	$\rho$ being the density associated to $\psi$.
\end{assumption}
\begin{remark}
	In the case $\alpha=2$ we may choose $e$ to be a positive constant.
\end{remark}
\begin{assumption}[Local Uncertainty]\label{ass:local_uncert_alpha}
	Given $W$, there exist $\alpha >0$ and constants $S_1,S_2>0$ such that
	for any finite cube $Q$, any $N \geq 1$ and all normalized $\psi \in H^1(\R^{dN})$ we have
	\begin{align*}
		(T+W)_{\psi}^{Q} 
		\geq \left\{\begin{array}{ll}
 				S_1\frac{\int_Q\rho^{1+2/d}}{(\int_Q\rho)^{2/d}} -S_2\frac{\int_Q\rho}{|Q|^{2/d}}, 
				&\quad \text{for }0<\alpha \leq 2,\\
				S_1\frac{\left(\int_Q\rho^{1+\alpha/d}\right)^{2/\alpha}}{(\int_Q\rho)^{2/\alpha+2/d-1}}
				-S_2\frac{\int_Q\rho}{|Q|^{2/d}}, 
				&\quad \text{for }\alpha > 2,
				\end{array} \right.
	\end{align*}
	where $\rho$ is the density associated to $\psi$.
\end{assumption}

\begin{remark}
	Assumption~\ref{ass:local_uncert_alpha} will typically be a consequence of
	Poincar\'e-Sobolev inequalities for the kinetic energy of the bosonic 
	wave function,
	since for $W \geq 0$ one has 
	the estimate $(T+W)_{\psi}^{Q} \geq T_{\psi}^Q$.
	In the case $0<\alpha\leq 2$, the inequality
	follows from \cite[Lemma 14]{Lundholm-Solovej:extended} with 
	the explicit constants
	\begin{align*}
		S_1 = C_d' \eps^{1+4/d}, \quad S_2 = C_d'\left(1+\left(\frac{\eps}{1-\eps}\right)^{1+4/d}\right),
	\end{align*}
	where $C_d' := \frac{\pi^2}{4}\frac{d^{2-2/d}}{(d+2)(d+4)}$ and arbitrary 
	$\eps \in (0,1)$.
	The more complicated case $\alpha > 2$ is discussed in the Appendix;
	see Proposition~\ref{prop:uncertainty_3geqd} - \ref{prop:uncertainty_1d}.
	Furthermore, note that our explicit constants given here are far from optimal.
\end{remark}	
	
\subsection{Preliminaries}
For the following lemmas we have $f(\gamma)=e(\gamma)$ in mind, but formulate them more generally.
\begin{lemma} \label{lem:general_concave}
	Let $f(\gamma)$ be a monotone increasing and concave function with $f(0)=0$. 
	Then $f$ has the following properties:
	\begin{enumerate}
		\item Monotonicity:
		\begin{align}\label{eq:e_monotonicity}
			f(\gamma_1) \leq f(\gamma_2), \quad \gamma_1 \leq \gamma_2.
		\end{align}
		\item With concavity and $f(0)=0$ one has
		\begin{align}\label{eq:e_concave_1}
			f(\eta\gamma) \leq \eta f(\gamma), \quad \eta \geq 1,
		\end{align}
		and
		\begin{align}\label{eq:e_concave_2}
			f(\eta\gamma) \geq \eta f(\gamma), \quad 0\leq \eta \leq 1.
		\end{align}
	\end{enumerate}
\end{lemma}

\begin{lemma}\label{lem:integral_bounds}
	Let $\tilde\rho := \int_Q \rho/|Q|$.
	For $\gamma$ as in \eqref{def:gamma} with $0 < \alpha \leq 2$,
	and $f$ satisfying the assumptions of Lemma~\ref{lem:general_concave},
	\begin{align} \label{eqn:general_bound}
		\int_Q f(\gamma(2/\rho)) \rho^{1+2/d}
		\leq f(\gamma(2/\tilde\rho)) \left( \int_Q \tilde\rho^{1+2/d} + \int_Q \rho^{1+2/d} \right),
	\end{align}
	hence, if $\int_Q \rho^{1+2/d} \le \Lambda\int_Q \tilde\rho^{1+2/d}$,
	\begin{align} \label{eqn:general_bound_constQ}
		\int_Q f(\gamma(2/\rho)) \rho^{1+2/d} 
		\le (1+\Lambda) f(\gamma(2/\tilde{\rho})) \frac{(\int_Q \rho)^{1+2/d}}{|Q|^{2/d}}.
	\end{align}
	On the other hand, for $\alpha \ge 2$,
	\begin{align} \label{eqn:general_bound2}
		\int_Q f(\gamma(2/\rho)) \rho^{1+2/d} 
		\le  f(\gamma(2/\tilde{\rho})) \left( \int_Q \tilde\rho^{1+2/d} 
		+\tilde\rho^{\frac{2-\alpha}{d}} \int_Q \rho^{1+\alpha/d} \right),
	\end{align}
	so if $\int_Q \rho^{1+\alpha/d} \le \Lambda\int_Q \tilde\rho^{1+\alpha/d}$, then
	\begin{align} \label{eqn:general_bound2_constQ}
		\int_Q f(\gamma(2/\rho)) \rho^{1+2/d} \le (1+\Lambda) 
		f(\gamma(2/\tilde{\rho})) \frac{(\int_Q \rho)^{1+2/d}}{|Q|^{2/d}}.
	\end{align}
\end{lemma}
\begin{proof}
	We write for the l.h.s. of \eqref{eqn:general_bound}
	\begin{align*}
		&\int_Q f\left(\gamma(\tilde\rho/\rho \cdot 2/\tilde\rho)\right) \rho^{1+2/d} \\
		&= \int_{\rho \le \tilde\rho} f\left((\tilde\rho/\rho)^{\frac{2-\alpha}{d}} 
		\gamma(2/\tilde\rho)\right) \rho^{1+2/d}
		+ \int_{\rho > \tilde\rho} f\left((\tilde\rho/\rho)^{\frac{2-\alpha}{d}} 
		\gamma(2/\tilde\rho)\right) \rho^{1+2/d} \\
		&\le \int_{\rho \le \tilde\rho} (\tilde\rho/\rho)^{\frac{2-\alpha}{d}} 
		f\left(\gamma(2/\tilde\rho)\right) \rho^{1+2/d}
		 + \int_{\rho > \tilde\rho} f\left(\gamma(2/\tilde\rho)\right) \rho^{1+2/d} \\
		&\le f\left(\gamma(2/\tilde\rho)\right) \left( \int_Q \tilde\rho^{\frac{2-\alpha}{d}} 
		\tilde\rho^{1+\alpha/d} + \int_Q \rho^{1+2/d} \right),
	\end{align*}
	where for the first integral we used concavity and for the second monotonicity.
	With the assumption of the lemma, this proves \eqref{eqn:general_bound_constQ}.
	
	We write for the l.h.s. of \eqref{eqn:general_bound2}
	\begin{align*}
		&\int_Q f\left(\gamma(\tilde\rho/\rho \cdot 2/\tilde\rho)\right) \rho^{1+2/d}\\
		&= \int_{\rho \le \tilde\rho} f\left((\rho/\tilde\rho)^{\frac{\alpha-2}{d}} 
		\gamma(2/\tilde\rho)\right) \rho^{1+2/d}
		+ \int_{\rho > \tilde\rho} f\left((\rho/\tilde\rho)^{\frac{\alpha-2}{d}} 
		\gamma(2/\tilde\rho)\right) \rho^{1+2/d}\\
		&\le \int_{\rho \le \tilde\rho} f\left(\gamma(2/\tilde\rho)\right) \rho^{1+2/d}
		 + \int_{\rho > \tilde\rho} (\rho/\tilde\rho)^{\frac{\alpha-2}{d}} 
		 f\left(\gamma(2/\tilde\rho)\right) \rho^{1+2/d} \\
		&\le f\left(\gamma(2/\tilde\rho)\right) \left( \int_Q \tilde\rho^{1+2/d} 
		+\tilde\rho^{\frac{2-\alpha}{d}} \int_Q \rho^{1+\alpha/d} \right),
	\end{align*}
	where for the first integral we used monotonicity and for the second concavity.
	With the assumption of the lemma we obtain \eqref{eqn:general_bound2_constQ}.
\end{proof}

\begin{lemma}\label{lem:f_quot_estimate}
	Given $|Q_B| >0$ and $1\leq j \leq k$, let $|Q_j| = 2^{d(k-j)}|Q_B|$. For $\gamma$ as in 
	\eqref{def:gamma} with $\alpha >0$, and $f$ satisfying the assumptions of 
	Lemma~\ref{lem:general_concave}, we have
	\begin{align*}
		\sum_{j=1}^k \frac{|Q_j|^{-2/d}f(\gamma(|Q_j|))}{|Q_B|^{-2/d}f(\gamma(|Q_B|))} 
		\leq \frac{1}{1-{2^{-\!\min\{\alpha,2\}}}}.
	\end{align*}
\end{lemma}
\begin{proof}
	Consider first the case $0<\alpha\leq 2$. 
	Using $|Q_j| = 2^{d(k-j)}|Q_B|$ in combination with concavity 
	\eqref{eq:e_concave_1} gives
	\begin{align*}
		f(\gamma(|Q_j|)) = f(2^{(2-\alpha)(k-j)}\gamma(|Q_B|)) 
		\leq 2^{(2-\alpha)(k-j)} f(\gamma(|Q_B|))
	\end{align*}
	for $j\leq k$. Thus, 
	\begin{align*}
		\sum_{j=1}^k \frac{|Q_j|^{-2/d}f(\gamma(|Q_j|))}{|Q_B|^{-2/d}f(\gamma(|Q_B|))}
		&\leq \sum_{j=1}^k \left(2^{d(k-j)}\right)^{-2/d}2^{(2-\alpha)(k-j)}\\
		&= \sum_{j=1}^k 2^{-\alpha(k-j)} = \frac{1-2^{-\alpha k}}{1-2^{-\alpha}} \leq \frac{1}{1-{2^{-\alpha}}}.
	\end{align*}
	In the case $\alpha \ge 2$ we can use the monotonicity
	\begin{align*}
		f(\gamma(|Q_j|)) = f(2^{-(\alpha-2)(k-j)}\gamma(|Q_B|)) \le f(\gamma(|Q_B|)),
	\end{align*}
	hence one obtains the upper bound $\sum_{j=1}^k 2^{-2(k-j)} \le (1-2^{-2})^{-1} = 4/3$.
\end{proof}

\subsection{Lieb-Thirring Type Bounds}
We are now in the position to prove our main theorem, for which we will however have to
assume that the function $e$ satisfying Assumption~\ref{ass:local_exclusion_alpha}
is bounded. 
For later convenience and clarity we introduce $\underline{e}_K$,
$$
	\gamma \mapsto \ue(\gamma) := \min\{e(\gamma),K\}, \quad K > 0,
$$
which replaces $e$ by a bounded monotone increasing and concave function.
This boundedness assumption can be relaxed,
but at the cost of only obtaining an estimate 
involving a local mean of the density $\rho$ 
(cf. \cite[Theorem~18 - 19]{Lundholm-Solovej:extended}).

\begin{theorem}[Lieb-Thirring inequality]\label{thm:L-T_global_alpha}
	Let $W$ satisfy Assumption~\ref{ass:local_exclusion_alpha} \& \ref{ass:local_uncert_alpha} 
	with an $\alpha > 0$ and $e$ replaced by $\ue$. 
	Then there exists a constant $C_{d,\alpha,K}>0$ given explicitly below, such that 
	for any $N \geq 1$ and all normalized $\psi \in H^1(\R^{dN})$, the total energy satisfies the estimate
	\begin{align*}
		\langle\psi,(\hat{T}+\hat{W})\psi\rangle 
		\geq C_{d,\alpha,K} \int_{\R^d} \ue(\gamma(2/\rho(\bx))) \rho(\bx)^{1+2/d} \,d
		\bx,
	\end{align*}
	where $\rho$ is the density associated to $\psi$.
\end{theorem}
\begin{proof}
	Note that by Assumption~\ref{ass:local_exclusion_alpha} and the
	Lebesgue differentiation theorem we have for almost every point 
	$\bx \in \R^d$ that 
	$$
		(T+W)_{\psi}(\bx) = \lim_{Q \to \bx} \frac{1}{|Q|}(T+W)_{\psi}^Q 
		\ \ \ge 0
	$$
	and hence $(T+W)_{\psi}^{Q} = \int_Q (T+W)_{\psi}(\bx) \,d\bx$ 
	is increasing with respect to $Q$.
	In the case $N \le 2$, we can therefore use Assumption~\ref{ass:local_uncert_alpha} in the
	limit $Q$ tending to $\R^d$ to obtain
	\begin{align*}
		\langle\psi,(\hat{T}+\hat{W})\psi\rangle 
		\geq \frac{S_1}{2^{2/d}} \int_{\R^d}\rho(\bx)^{1+2/d}\,d\bx,\end{align*}
	(where we used H\"older's inequality in the case $\alpha >2$)
	and hence the desired estimate with the constant 
	$C_{d,\alpha,K}$ being $c_0 := \frac{S_1}{2^{2/d}K}$. We can henceforth assume 
	that $N\geq 3$. 
	
	Consider now any cube $Q_0 \subset \R^d$ with $\int_{Q_0}\rho \geq 2$. 
	Following \cite{Lundholm-Solovej:anyon, Lundholm-Solovej:extended}, 
	we split $Q_0$ into disjoint sub-cubes 
	$Q_A,Q_B$, organized in a tree $\mathbb{T}$, such that
	\begin{align*}
		&\textrm{on }Q_A: \quad 0 \leq \int_{Q_A}\rho < 2,\\
		&\textrm{on }Q_B: \quad 2 \leq \int_{Q_B}\rho < 2^{d+1},
	\end{align*}
	and on any such sub-cube $Q$ we define the local mean density 
	$\tilde\rho|_Q := \int_Q \rho/|Q|$.
	Note that the structure of the tree is such that at least one $B$-cube 
	can be found among the $2^d$ top-level
	leaves of every branch of the tree (see Fig.~3 in \cite{Lundholm-Solovej:anyon}). 
	We now treat the following two cases separately.
	
	\paragraph{Case $0<\alpha \leq 2$:}
	Consider first any such A- or B-cube $Q$ on which the density is very non-constant
	in the sense that
	\begin{align}\label{eq:nonc_density_small_alpha}
		\int_{Q}\rho^{1+2/d} > \Lambda\int_{Q}\tilde{\rho}^{1+2/d} 
		= \Lambda\frac{(\int_{Q}\rho)^{1+2/d}}{|Q|^{2/d}},
	\end{align}
	where $\Lambda > 0$ is a sufficiently large constant to be chosen below.
	Using Assumption~\ref{ass:local_uncert_alpha} with $\alpha \leq 2$,
	the bound \eqref{eq:nonc_density_small_alpha}, 
	and that $\int_Q \rho < 2^{d+1}$,
	we then have
	\begin{align}
		(T+W)_{\psi}^Q &\geq
		S_1\frac{\int_{Q}\rho^{1+2/d}}{(\int_{Q}\rho)^{2/d}} 
		- S_2\frac{\int_{Q}\rho}{|Q|^{2/d}}\nonumber\\
		&\geq \frac{S_1}{2}\frac{\int_{Q}\rho^{1+2/d}}{(\int_{Q}\rho)^{2/d}}
		+\left(\frac{S_1}{2}\Lambda-S_2\right)\frac{\int_Q\rho}{|Q|^{2/d}}\nonumber\\
		&\geq c_1 \int_{Q} K \rho^{1+2/d}
		\geq c_1 \int_{Q}\ue(\gamma(2/\rho))\rho^{1+2/d},
		\label{eq:bound_cube_nonc_small_alpha}
	\end{align}
	with $c_1 := 2^{-3-2/d}S_1/K$ if $\Lambda := 2S_2/S_1$.

	On B-cubes $Q_B$ we have sufficiently many particles to
	use Assumption~\ref{ass:local_exclusion_alpha}:
	\begin{align}\label{eq:bound_cube_exclusion_small_alpha}
		(T+W)_{\psi}^{Q_B} 
		\ge \frac{1}{2}\ue(\gamma(|Q_B|))
		\frac{ \int_{Q_B}\rho -1 }{|Q_B|^{2/d}}.
	\end{align}
	By \eqref{eq:bound_cube_nonc_small_alpha}, we can restrict to the case of 
	nearly constant density (the converse of \eqref{eq:nonc_density_small_alpha}).
	The relation
	$\gamma(|Q_B|) = \gamma(\int_{Q_B}\rho/\tilde{\rho}) \ge \gamma(2/\tilde\rho)$
	and monotonicity of $\ue$, 
	together with the inequality $x-1 \geq (2^{d+1}-1)/(2^{d+1})^p x^p$ 
	when $2 \leq x:= \int_{Q_B}\rho < 2^{d+1}$ and $p = 1+2/d$, 
	produces the further lower bound
	\begin{align*}
		(T+W)_{\psi}^{Q_B} 
		\ge \frac{1}{2}\ue(\gamma(2/\tilde\rho)) 
		\frac{2^{d+1}-1}{2^{(d+1)(1+2/d)}}
		\frac{(\int_{Q_B}\rho)^{1+2/d}}{|Q_B|^{2/d}},
	\end{align*}
	and finally, by means of \eqref{eqn:general_bound_constQ} 
	in Lemma \ref{lem:integral_bounds}
	with $f=\ue$
	and the converse of \eqref{eq:nonc_density_small_alpha}
	we obtain
	\begin{align}\label{eq:bound_cube_const_B_small_alpha}
		(T+W)_{\psi}^{Q_B} 
		\ge c_2 \int_{Q_B} \ue(\gamma(2/\rho)) \rho^{1+2/d},
	\end{align}
	with $c_2 :=2^{-1-(d+1)(1+2/d)}(2^{d+1}-1)/(1+\Lambda)$.

	It remains then to consider all A-cubes $Q_A$ with nearly constant density.
	Again, using \eqref{eqn:general_bound_constQ} with $f=\ue$ we have
	\begin{align*}
		\int_{Q_A} \ue(\gamma(2/\rho)) \rho^{1+2/d} 
		&\le (1+\Lambda) \ue(\gamma(2/\tilde{\rho})) 
		\frac{(\int_{Q_A} \rho)^{1+2/d}}{|Q_A|^{2/d}}\\
		&\le (1+\Lambda) 2^{1+2/d} \frac{\ue(\gamma(|Q_A|))}{|Q_A|^{2/d}},
	\end{align*}
	where in the second step we once more used concavity, and $\int_{Q_A}\rho<2$:
	\begin{align*}
		\ue\left(\gamma\left((2/\textstyle{\int_{Q_A}\rho})|Q_A|\right)\right) 
		&= \ue\left((2/\textstyle{\int_{Q_A}\rho})^{(2-\alpha)/d} \gamma(|Q_A|)\right)\\
		&\le \left(2/\textstyle{\int_{Q_A}\rho}\right)^{(2-\alpha)/d} \ue(\gamma(|Q_A|)).
	\end{align*}
	Hence, applying Lemma~\ref{lem:f_quot_estimate}
	for the collection of all such $Q_A$ associated to a cube $Q_B$ 
	at some level $k$ in the tree,
	with maximally $2^d-1$ such A-cubes at each level $j \le k$
	(cf. \cite{Lundholm-Solovej:anyon}), 
	\begin{align}
		\sum_{Q_A} \int_{Q_A} \ue(\gamma(2/\rho)) \rho^{1+2/d}
		&\le (2^d-1)(1+\Lambda) 2^{1+2/d}\sum_{j=1}^k 
		\frac{\ue(\gamma(|Q_j|))}{|Q_j|^{2/d}}\nonumber\\
		&\le c_3\frac{\ue(\gamma(|Q_B|))}{|Q_B|^{2/d}},
		\label{eq:bound_cube_const_A_small_alpha}
	\end{align}
	with $c_3 := 2^{1+2/d}(2^d-1)(1+\Lambda)/(1-2^{-\alpha})$. 
	This quantity is (after rescaling by $4c_3$) covered by half of the energy 
	$(T+W)^{Q_B}$ given by \eqref{eq:bound_cube_exclusion_small_alpha},
	leaving the other half for the bounds \eqref{eq:bound_cube_nonc_small_alpha} 
	or \eqref{eq:bound_cube_const_B_small_alpha} on $Q_B$.
	
	Thus, summing up the integrals \eqref{eq:bound_cube_nonc_small_alpha}, 
	\eqref{eq:bound_cube_const_B_small_alpha} and 
	\eqref{eq:bound_cube_const_A_small_alpha}, we have
	\begin{align*}
		(T+W)_{\psi}^{Q_0} \ge C_{\alpha,d,K} \int_{Q_0} \ue(\gamma(2/\rho)) \rho^{1+2/d}
	\end{align*} 
	with
	\begin{align*}
		C_{d,\alpha,K} := \min\left\{c_0,\frac{c_1}{2},\frac{c_2}{2},\frac{1}{4c_3}\right\}.
	\end{align*}

	\paragraph{Case $\alpha > 2$:}
	We proceed as in the previous case, although the condition $\alpha > 2$ 
	changes the roles of monotonicity and concavity accordingly, and 
	furthermore demands a stronger version of the uncertainty principle 
	(see Assumption~\ref{ass:local_uncert_alpha}).
	
	Consider first any A- or B-cube $Q$ on which the density is very
	non-constant in the sense that
	\begin{align}\label{eq:nonc_density_large_alpha}
		\int_{Q}\rho^{1+\alpha/d} > \Lambda\int_{Q}\tilde{\rho}^{1+\alpha/d} 
		= \Lambda\frac{(\int_{Q}\rho)^{1+\alpha/d}}{|Q|^{\alpha/d}},
	\end{align}
	where $\Lambda > 0$ is a sufficiently large constant to be chosen below.
	Using Assumption~\ref{ass:local_uncert_alpha} (with $\alpha>2$), 
	followed by H\"older's inequality and \eqref{eq:nonc_density_large_alpha},
	and finally $\int_Q \rho < 2^{d+1}$, we then have
	\begin{align}
		(T+W)_{\psi}^{Q}
		&\geq S_1\frac{\left(\int_Q\rho^{1+\alpha/d}\right)^{2/\alpha}}{(\int_Q\rho)^{2/\alpha+2/d-1}}
		-S_2\frac{\int_Q\rho}{|Q|^{2/d}}\nonumber\\
		&\geq \frac{S_1}{2}\frac{\int_{Q}\rho^{1+2/d}}{(\int_{Q}\rho)^{2/d}}
		+\left(\frac{S_1}{2}\Lambda^{2/\alpha}-S_2\right)\frac{\int_Q\rho}{|Q|^{2/d}}\nonumber\\
		&\geq c_4 \int_{Q} K \rho^{1+2/d}
		\geq c_4 \int_{Q}\ue(\gamma(2/\rho))\rho^{1+2/d},
		\label{eq:bound_cube_nonc_large_alpha}
	\end{align}
	with $c_4:= 2^{-3-2/d}S_1/K$ if $\Lambda := (2S_2/S_1)^{\alpha/2}$.

	On B-cubes $Q_B$ we have sufficiently many particles to
	use Assumption~\ref{ass:local_exclusion_alpha}:
	\begin{align}\label{eq:bound_cube_exclusion_2}
		(T+W)_{\psi}^{Q_B} 
		\ge \frac{1}{2}\ue(\gamma(|Q_B|))
		\frac{ \int_{Q_B}\rho -1 }{|Q_B|^{2/d}}.
	\end{align}
	By \eqref{eq:bound_cube_nonc_large_alpha}, we can restrict to the case of 
	nearly constant density (the converse of \eqref{eq:nonc_density_large_alpha}).
	The relation
	\begin{align*}
		\textstyle{\gamma(|Q_B|) = \gamma\left(\int_{Q_B}\rho/2 \cdot 2/\tilde{\rho}\right) 
		= \left(\int_{Q_B}\rho/2\right)^{-(\alpha-2)/d} \gamma(2/\tilde\rho)}
	\end{align*}
	and concavity \eqref{eq:e_concave_2}, 
	together with the inequality $x-1 \geq (2^{d+1}-1)/(2^{d+1})^p x^p$ 
	when $2 \leq x:= \int_{Q_B}\rho < 2^{d+1}$ and $p=1+\alpha/d$, 
	produces the further lower bound
	\begin{align*}
		(T+W)_{\psi}^{Q_B} 
		&\ge \frac{1}{2} \left( \frac{\int_{Q_B}\rho}{2} \right)^{-(\alpha-2)/d} 
		\ue(\gamma(2/\tilde\rho)) 
		\frac{2^{d+1}-1}{2^{(d+1)(1+\alpha/d)}}
		\frac{(\int_{Q_B}\rho)^{1+\alpha/d}}{|Q_B|^{2/d}} \\
		&= 2^{-(2+\alpha+d+2/d)}(2^{d+1}-1)\ue(\gamma(2/\tilde\rho)) 
		\frac{(\int_{Q_B}\rho)^{1+2/d}}{|Q_B|^{2/d}},
	\end{align*}
	and finally, by means of \eqref{eqn:general_bound2_constQ} 
	in Lemma \ref{lem:integral_bounds}
	with $f=\ue$, which requires
	the converse of \eqref{eq:nonc_density_large_alpha}:
	\begin{align}\label{eq:bound_cube_const_B_large_alpha}
		(T+W)_{\psi}^{Q_B} 
		\ge c_5 \int_{Q_B} \ue(\gamma(2/\rho)) \rho^{1+2/d},
	\end{align}
	where $c_5:= 2^{-(2+\alpha+d+2/d)}(2^{d+1}-1)/(1+\Lambda)$.

	It remains to consider all A-cubes with nearly constant density.
	Again, using \eqref{eqn:general_bound2_constQ} with $f=\ue$ 
	we have
	\begin{align*}
		\int_{Q_A} \ue(\gamma(2/\rho)) \rho^{1+2/d} 
		&\le (1+\Lambda) \ue(\gamma(2/\tilde{\rho})) 
		\frac{(\int_{Q_A} \rho)^{1+2/d}}{|Q_A|^{2/d}}\\
		&\le (1+\Lambda) 2^{1+2/d} \frac{\ue(\gamma(|Q_A|))}{|Q_A|^{2/d}},
	\end{align*}
	where we here used monotonicity
	\begin{align*}
		\textstyle{\ue\left(\gamma\left((2/\int_{Q_A}\rho)|Q_A|\right)\right) 
		= \ue\left((2/\int_{Q_A}\rho)^{-(\alpha-2)/d} \gamma(|Q_A|)\right)
		\le \ue(\gamma(|Q_A|))}.
	\end{align*}
	Hence, applying Lemma~\ref{lem:f_quot_estimate}
	for the collection of all such $Q_A$ associated to a cube $Q_B$ 
	at some level $k$ in the tree,
	\begin{align}
		\sum_{Q_A} \int_{Q_A} \ue(\gamma(2/\rho)) \rho^{1+2/d}
		&\le (2^d-1)(1+\Lambda) 2^{1+2/d}\sum_{j=1}^k 
		\frac{\ue(\gamma(|Q_j|))}{|Q_j|^{2/d}}\nonumber\\
		&\le c_6\frac{\ue(\gamma(|Q_B|))}{|Q_B|^{2/d}},
		\label{eq:bound_cube_const_A_large_alpha}
	\end{align}
	where $c_6 := 2^{1+2/d}(2^d-1)(1+\Lambda)/(1-2^{-2})$.
	This quantity is (after rescaling by $4c_6$) covered by half of the energy 
	$(T+W)^{Q_B}$ given by \eqref{eq:bound_cube_exclusion_2},
	leaving the other half for the bounds \eqref{eq:bound_cube_nonc_large_alpha} 
	or \eqref{eq:bound_cube_const_B_large_alpha} on $Q_B$.

	Thus, summing up the integrals \eqref{eq:bound_cube_nonc_large_alpha}, 
	\eqref{eq:bound_cube_const_B_large_alpha} and 
	\eqref{eq:bound_cube_const_A_large_alpha}, 
	we have
	\begin{align*}
		(T+W)_{\psi}^{Q_0} \ge C_{\alpha,d,K}\int_{Q_0} \ue(\gamma(2/\rho)) \rho^{1+2/d}
	\end{align*} 
	with
	\begin{align*}
		C_{d,\alpha,K} := \min\left\{c_0,\frac{c_4}{2}, \frac{c_5}{2},\frac{1}{4c_6}\right\}.
	\end{align*}
	
	Finally, we can let $Q_0$ tend to the whole of $\R^d$ using monotone convergence.
\end{proof}
\begin{remark}
	In the case $0<\alpha \leq 2$, given the explicit expression for $S_1$ and $S_2$ 
	in the remark after Assumption~\ref{ass:local_uncert_alpha}, and 
	taking $\eps = 1/2$, one can compute
	all constants in $C_{d,\alpha,K}$ explicitly, giving
	$c_0 =2^{-1-6/d}C_d'/K$, 
	$\Lambda = 2^{3+4/d}$,
	$c_1 = 2^{-4-6/d}C_d'/K$, 
	$c_2 = 2^{-4-d-2/d}(2^{d+1}-1)/(1+2^{3+4/d})$ and
	$c_3 = 2^{1+2/d}(2^d-1)(1+2^{3+4/d})/(1-2^{-\alpha})$.
\end{remark}
\begin{remark}
	When $K \to \infty$, we have $C_{d,\alpha,K} \leq  c_1/2 \to 0$ ($\alpha \leq 2$) and
	$C_{d,\alpha,K} \leq c_4/2 \to 0$ ($\alpha >2$).
\end{remark}

For Theorem~\ref{thm:L-T_global_alpha} we required that the function $e$ from 
Assumption~\ref{ass:local_exclusion_alpha} is replaced by a
bounded function $\underline{e}_K$.
As already remarked, this restriction can be dropped, but at the cost of only 
obtaining an estimate involving the local mean $\tilde{\rho}$ 
(as defined in the proof) of the density $\rho$. Namely, exclusion 
with unbounded strength cannot be matched by uncertainty to produce a uniform 
Lieb-Thirring type inequality (cf. \cite[Theorem~18 - 19]{Lundholm-Solovej:extended}).
The approach involving the local mean has the further advantage that 
Assumption~\ref{ass:local_uncert_alpha} is not required, 
i.e. it relies on local exclusion alone.

\section{Applications}\label{sec:applications}
Here we consider some important examples for which concrete bounds of the form \eqref{eq:two_particle_bound} can be obtained, hence resulting 
in corresponding Lieb-Thirring type bounds as corollaries 
of Theorem~\ref{thm:L-T_global_alpha}.

\subsection{The Lieb-Liniger Model} \label{sec:Lieb-Liniger}
The Lieb-Liniger model (see \cite{Lieb-Liniger:63}) describes $N$ bosons in one dimension 
with pairwise point interactions. The interaction Hamiltonian is given by
\begin{align*}
	\hat{T} + \hat{W} 
	= -\sum_{j=1}^N\frac{\partial^2}{\partial x_j^2} + 4\eta\sum_{1\leq j<k\leq N}\delta(x_j-x_k).
\end{align*}
For repulsive interactions we have the zero-range pair 
potential $W(x) = 4\eta \delta(x)$ with $\eta \geq 0$, and we obtain 
Assumption~\ref{ass:local_exclusion_alpha}  from \cite[Lemma~13]{Lundholm-Solovej:extended} 
with $e(\gamma) = 4\xi_{\LL}(\gamma)^2 = \underline{e}_{K=\pi^2}(\gamma)$ and
$\gamma(|Q|) := \eta|Q|$ (hence $\alpha = 1$, $\calpha = \eta$). The bounded 
concave function $\xi_{\LL}(\gamma)$ 
is defined as the smallest non-negative solution to the equation $\xi\tan\xi = \gamma$.
Furthermore, Assumption~\ref{ass:local_uncert_alpha} holds with the parameters 
given in the remark following the assumption.

\begin{theorem}[Lieb-Thirring inequality for Lieb-Liniger]\label{thm:L-T_LL}
	There exists a constant $C_{\LL} > 3 \cdot 10^{-5}$
	such that for any $\eta \geq 0$, any $N\geq 1$ and all normalized $\psi \in H^1(\R^{N})$
	the total energy satisfies the estimate
	\begin{align*}
		\langle\psi,(\hat{T}+4\eta\sum_{j<k}\delta(x_j-x_k))\psi\rangle 
		\geq C_{\LL} \int_{\R}\xi_{\LL}(2\eta/\rho(x))^2 \rho(x)^{3}\,dx,
	\end{align*}
	where $\rho$ is the density associated to $\psi$.
\end{theorem}
\begin{proof}
	Given the above parameters, the statement of Theorem~\ref{thm:L-T_global_alpha} holds with 
	the constant $C_{1,1,\pi^2} = c_1/2 = 1/122880$ (see the remark following the theorem), hence 
	$C_{\LL} = 4C_{1,1,\pi^2} > 3\cdot 10^{-5}$.
\end{proof}
\begin{remark}
	The above result should be compared with 
	\cite[Theorem~17]{Lundholm-Solovej:extended}, taking the 
	difference in the conventions for defining the kinetic energy into account. 
	Observe that the methods 
	employed here eliminate the use of the Hardy-Littlewood maximal function.
\end{remark}

\subsection{Homogeneous Potentials}
In this subsection we consider the family of potentials of the form
\begin{align*}
	W_{\beta}(\bx) = \frac{W_{0}}{|\bx|^{\beta}}
\end{align*}
in arbitrary dimensions $d$, where $\beta >0$ and $W_{0}$ is a positive constant.

\subsubsection{Inverse-Square Interaction} \label{sec:inverse-square}
In the case $\beta=2$, the interaction takes the form $W_2(\bx) = W_{0}|\bx|^{-2}$ 
and the potential now scales in the same way as the kinetic energy. 
This implies that the (scale-normalized) two-particle energy $e_2(|Q|;W_2)$ 
is constant as a function of $|Q|$ (and non-zero, see 
Section~\ref{sssec:hom_elementary_est}),
\begin{align*}
	e_2(|Q|;W_2) = e_2(W_{0}) =: e,
\end{align*}
and thus $\gamma(|Q|)$ is constant (i.e. $\alpha=2$). Assumption~\ref{ass:local_exclusion_alpha} 
is
then a direct consequence of  Theorem~\ref{thm:local_ex_TW}, 
\begin{align*}
	(T+W_2)_{\psi}^{Q} \geq \frac{1}{2}\frac{e}{|Q|^{2/d}}\left(\int_{Q}\rho(\bx)\,d\bx - 1\right)_{+},
\end{align*}
and we obtain the following corollary from Theorem~\ref{thm:L-T_global_alpha}.
\begin{theorem}\label{thm:square_int_L-T}
	Given $W_2(\bx) = W_{0}|\bx|^{-2}$ with $W_0>0$,
	 there exists $e=e_2(W_0)>0$ and a positive constant $C_{d,2,e}$ such that for any $N 
	\geq 1$ and all
	normalized $\psi \in H^1(\R^{dN})$ one has the Lieb-Thirring inequality
	\begin{align*}
		\langle\psi, (\hat{T}+\hat{W}_{2})\psi\rangle 
		\geq C_{d,2,e}\, e\int_{\R^d}\rho(\bx)^{1+2/d}\,d\bx,
	\end{align*}
	where $\rho$ is the density associated to $\psi$.
\end{theorem}

The above result should be compared to the situation where one considers free fermions instead 
of interacting  bosons. A fermionic wave function $\psi$ with particle number $n \geq 2$ can be 
shown to satisfy\footnote{Alternatively, one can in this regard view fermions as bosons with a local 
repulsive inverse-square interaction; see e.g. \cite[Theorem~2.8]{H-O2-Laptev-Tidblom:08}, 
and cf. the local approach to exclusion in \cite{Lundholm-Solovej:anyon}.}, 	
only using the antisymmetry of the wave function (cf. \cite[Lemma~5]{Dyson-Lenard:67}),
\begin{align}\label{eq:dl_fermion_bound}
	\int_{Q^n}\sum_{j=1}^n|\nabla_j\psi|^2 d\sx 
	\geq (n-1)\frac{\pi^2}{|Q|^{2/d}}\int_{Q^n}|\psi|^2\,d\sx.
\end{align}
This leads, using ideas as in the proof of Theorem~\ref{thm:local_ex_TW} 
(cf. \cite{Lundholm-Solovej:anyon, Lundholm-Solovej:extended}), 
to an estimate
\begin{align*}
	T_{\psi}^Q \geq \frac{\pi^2}{|Q|^{2/d}}\left(\int_{Q}\rho(\bx)\,d\bx - 1\right)_{+},
\end{align*}
and hence to a Lieb-Thirring inequality for fermions. An inequality similar to 
\eqref{eq:dl_fermion_bound} was used by Dyson and Lenard to prove the stability 
of fermionic matter 
(see \cite{Dyson-Lenard:67, Dyson-Lenard:68,Dyson:68, Lenard:73, Fefferman:83}), 
however without the more generally applicable tool of Lieb-Thirring inequalities.
Note that, applying Theorem~\ref{thm:square_int_L-T} and following the usual route to stability 
(see e.g. \cite{Lieb-Seiringer:10}), it follows that also a system of bosons with 
Coulomb interactions and inverse-square repulsive cores is 
thermodynamically stable.

\subsubsection{Elementary Estimate for Homogeneous Potentials}\label{sssec:hom_elementary_est}
We have the following very elementary bound for the two-particle interaction,
\begin{align}\label{eq:w_elementary_bound_1}
	e_2(|Q|;W_{\beta}) 
	\geq |Q|^{2/d}\inf_{(\bx_1,\bx_2)\in Q^2}W_{\beta}(\bx_1-\bx_2) 
	= d^{-\beta/2}W_{0}|Q|^{(2-
	\beta)/d}.
\end{align}
Defining $\gamma(|Q|) := W_{0}|Q|^{(2-\beta)/d}$ (i.e. $\alpha = \beta$, 
$\calpha=W_{0}$), the above estimate translates to
\begin{align}\label{eq:w_elemetary_bound_2}
	e_2(|Q|;W_{\beta}) \geq d^{-\beta/2}\gamma =: e(\gamma),
\end{align}
hence Assumption~\ref{ass:local_exclusion_alpha} is satisfied
by Theorem~\ref{thm:local_ex_TW}.
Observe that for $\beta \neq 2$, the function $e(\gamma)$ is unbounded. If we however
decide to bound $e(\gamma)$ from below by $\underline{e}_{K=\pi^2}(\gamma) 
= \min\{d^{-\beta/2}\gamma,\pi^2\}$, we are in the position to obtain global 
estimates from Theorem~\ref{thm:L-T_global_alpha},
whenever Assumption~\ref{ass:local_uncert_alpha} holds (see the Appendix).

\begin{theorem}\label{thm:L-T_global_elementary}
	Let $W_{\beta}(\bx) = W_{0}|\bx|^{-\beta}$ with $W_0 >0$. For $d=1,2$ we
	allow
	$0<\beta<\infty$, whereas for $d \geq 3$ we require $0<\beta\leq 2d/(d-2)$. Then
	there exists a constant $C_{d,\beta,\pi^2} > 0$ such that for any
	$N \geq 1$ and all normalized $\psi \in H^1(\R^{dN})$ one has
	\begin{multline*}
		\langle\psi, (\hat{T}+\hat{W}_{\beta})\psi\rangle\\
		\geq C_{d,\beta,\pi^2}\int_{\R^d}\min\left\{d^{-\beta/2}2^{(2-\beta)/d}W_{0}
		\rho(\bx)^{1+\beta/d},\pi^2\rho(\bx)^{1+2/d}\right\}\,d\bx,
	\end{multline*}
	where $\rho$ is the density associated to $\psi$.
\end{theorem}
\begin{remark}
	The special case $\beta = 2$ corresponds to the class of inverse-square interactions 
	and was discussed above. 
	In this case $\gamma(|Q|) = W_{0}$ and the elementary bound gives 
	$e(\gamma) = d^{-1}W_{0}$, again a constant. 
\end{remark}
\begin{remark}
	Note that if $W_{0} = a^{\beta-2}$, where $a > 0$ is a constant, 
	then $W_\beta(\bx)$ converges pointwise to
	a hard-sphere potential $W_a^{\hs}(\bx)$ in the limit $\beta \to \infty$.
\end{remark}

\subsection{Hard-Sphere Interaction}
The hard-sphere interaction of range $a$ (in $d=3$) corresponds to the potential
\begin{align*}
	W_{a}^{\hs}(\bx)= \left\{ \begin{array}{ll}
 					+\infty, &\quad |\bx| \leq a,\\
					0, &\quad |\bx|>a.
					\end{array} \right.
\end{align*}
Formally, $W_{a}^{\hs}$ is realized by introducing appropriate Dirichlet boundary conditions. 
The many-particle wave function $\psi(\bx_1,\dots, \bx_N)$ is required to vanish as soon as 
$|\bx_i-\bx_j| \le a$, $1\leq i < j\leq N$. We denote this subspace of $H^1(\R^{3N})$ by 
$H_a^1(\R^{3N})$. 

The two-particle energy $e_2(|Q|;W_{a}^{\hs})$ then becomes
\begin{align*}
	e_2(|Q|;W_a^{\hs}) = |Q|^{2/3}\inf_{\substack{\psi \in H_a^1(\R^{6})\\ \int_{Q^2}|\psi|^2=1}}\,\int_{Q^2}\left[|
	\nabla_1\psi|^2 + |\nabla_2\psi|^2\right]\,d\bx_1d\bx_2.
\end{align*}
We recall that the scattering length of $W_{a}^{\hs}$ is given by the range of the hard sphere, namely $a$.

\begin{proposition} \label{prop:hard_sphere_e_est}
	For $W_{a}^{\hs}$, the two-particle energy satisfies the estimate
	\begin{align*}
		e_2(|Q|;W_{a}^{\hs}) \geq \frac{2}{\sqrt{3}}\frac{a}{|Q|^{1/3}}.
	\end{align*}
\end{proposition}
\begin{proof}
	The key to proving the proposition is a lemma from \cite{Dyson:57}, 
	which we restate here for convenience.
	Let $G(t) \geq 0$ be any function defined for $0<t<\infty$ with
	\begin{align*}
		I := \int_0^{\infty}G(t)\,dt < \infty.
	\end{align*}
	\begin{lemma}[Dyson's Lemma]
		Let $\psi(\bx)$ be any function of the space point $\bx$, 
		defined in a region $\Omega \subset \R^3$. 
		Suppose that $\Omega$ is ``star-shaped" with respect to $0$ and $\psi(\bx) = 0$ for $|\bx|\le a$, 
		then 
		\begin{align*}
			\int_{\Omega}|\nabla\psi(\bx)|^2\,d\bx 
			\geq \frac{3a}{I}\int_{\Omega}G(|\bx|^3)|\psi(\bx)|^2\,d\bx.
		\end{align*}
	\end{lemma}
	
	To use the above lemma, we first define
	\begin{align*}
		G(t) := \left\{ \begin{array}{ll}
 					1, &\quad \textrm{if }t \leq 3^{3/2}|Q|,\\
					0, &\quad \textrm{otherwise},
					\end{array} \right.
	\end{align*}
	with $I = 3^{3/2}|Q|$. If we fix $\bx_2 \in Q$ and define 
	$\Omega := Q - \bx_2$, then Dyson's lemma gives
	\begin{align*}
		\int_Q |\nabla_1\psi|^2\,d\bx_1 &\geq \frac{3a}{I}\int_Q G(|\bx_1-\bx_2|^3) |\psi|^2\,d\bx_1\\
		&= \frac{3}{3^{3/2}}\frac{a}{|Q|}\int_Q |\psi|^2\,d\bx_1.
	\end{align*}
	The above can be repeated for the term $\int_Q|\nabla_2\psi|^2\,d\bx_2$ to yield a similar result. 
	Adding both terms and integrating over the respective variables, one finds
	\begin{align*}
		\int_{Q^2}\left[|\nabla_1\psi|^2 + |\nabla_2\psi|^2\right]\,d\bx_1d\bx_2 
		\geq \frac{2}{\sqrt{3}}\frac{a}{|Q|}\int_{Q^2}|\psi|^2\,d\bx_1d\bx_2,
	\end{align*}
	from which the Proposition follows immediately.
\end{proof}
Hence, if we set $\gamma(|Q|) := a|Q|^{-1/3}$ (i.e. $\alpha = 3$, $\calpha = a$), 
then
\begin{align*}
	e_2(|Q|;W_{a}^{\hs}) \geq \frac{2}{\sqrt{3}}\gamma =: e(\gamma),
\end{align*}
and Assumption~\ref{ass:local_exclusion_alpha} then follows directly from 
Theorem~\ref{thm:local_ex_TW}.
We furthermore bound $e(\gamma)$ from below by the bounded function 
$\underline{e}_{K=\pi^2}(\gamma)$
in order to apply Theorem~\ref{thm:L-T_global_alpha}.
Also note that Assumption~\ref{ass:local_uncert_alpha} holds for this value of $\alpha$ by virtue of 
Proposition~\ref{prop:uncertainty_3geqd} in the Appendix.

\begin{theorem}\label{thm:L-T_HS}
	Let $W_{a}^{\hs}$ denote the hard-sphere interaction of range $a>0$. 
	There exists a positive constant $C_{\hs}$ independent of $a$ 
	such that for any $N \geq 1$
	and all normalized $\psi \in H_a^1(\R^{3N})$ the total energy satisfies the estimate
	\begin{align*}
		\langle\psi,(\hat{T}+\hat{W}_{a}^{\hs})\psi\rangle 
		\geq C_{\hs}\int_{\R^3} \min\left\{\frac{2^{2/3}}{\sqrt{3}}
		a\rho(\bx)^2,\pi^2 \rho(\bx)^{5/3}\right\}\,d\bx,
	\end{align*}
	where $\rho$ is the density of $\psi$.
\end{theorem}
\begin{proof}
	In the above setting, we can apply Theorem~\ref{thm:L-T_global_alpha} with 
	$d=3$, $\alpha = 3$ and $K = \pi^2$. The statement of the theorem then holds with 
	$C_{\hs} :=  C_{3,3,\pi^2}$, which involves the constants
	from Proposition~\ref{prop:uncertainty_3geqd}. Using the explicit lower bound for
	$S_3$ given in the remark following this proposition, we compute that
	$C_{3,3,\pi^2} = c_4/2 > 4.5\cdot 10^{-6}$.
\end{proof}
\begin{remark}
	For small $\rho$, the appearance of the quadratic expression in 
	$\rho$ is natural and should be compared with 
	the exact expression obtained in \cite{Lieb-Yngvason:98} for the energy per 
	unit volume $e_0$ for the dilute 
	Bose gas in three dimensions, namely
	\begin{align*}
		\lim_{a^3\rho \to 0}\frac{e_0}{4\pi a\rho^2} = 1.
	\end{align*}
\end{remark}

\subsection{Estimates for Homogeneous Potentials in Terms of the Scattering Length in 3D}
In this subsection, we will prove a more refined estimate than 
the elementary bound (Theorem~\ref{thm:L-T_global_elementary})
for the potential 
$W_{\beta}(\bx) = W_{0}|\bx|^{-\beta}$ in $d=3$ in terms of 
its scattering length. Such an estimate is possible provided $\beta > 3$, 
for which the scattering length is finite.

From \eqref{eq:scattering_l_homogeneous} in the Appendix we obtain the scattering length for $W_\beta$:
\begin{align}\label{def:hom_scat_l}
	a_{\beta} = \Lambda_{\beta} \left(\frac{W_{0}}{2}\right)^{1/(\beta-2)}, 
	\quad \text{where }\Lambda_{\beta} := \frac{\Gamma\left(\frac{\beta-3}
	{\beta-2}\right)}{\Gamma\left(\frac{\beta-1}{\beta-2}\right)}\left(\frac{2}{\beta-2}\right)^{2/(\beta-2)}.
\end{align}
\begin{proposition}\label{prop:hom_pot_e_est}
	For $W_{\beta}$ with $\beta > 3$, the two-particle energy satisfies
	\begin{align*}
		e_2(|Q|;W_{\beta}) \geq \frac{\zeta}{\Lambda_{\beta}}\frac{a_{\beta}}{|Q|^{1/3}},
	\end{align*}
	where $\zeta :=(1-\tanh1)/\sqrt{3} > 0.137$.
\end{proposition}
\begin{proof}
	Our strategy is to replace $W_{\beta}(\bx)$ by a cut-off version
	\begin{align*}
		W_{\beta}^R(\bx) := \frac{W_{0}}{R^{\beta}}\chi_{B_{0,R}}(\bx),
	\end{align*}
	which has a finite range $R$ and a related scattering length $a_{\beta}^R$ (see \eqref{eq:scattering_l_w_cut} 
	in the Appendix), where we choose
	\begin{align}
		R:= \min\{\delta a_{\beta},\sqrt{3}|Q|^{1/3}\}
	\end{align}
	and $\delta > 0$ is to be determined later. 
	
	We note that because of the special form of the potential $W_{\beta}$, the elementary bound will be 
	sufficient to obtain the desired estimate on sufficiently small cubes compared to the 
	scattering length. This is accounted for by the form of the potential $W_{\beta}^R$. For large cubes 
	however, we will use the comparatively small range of $W_{\beta}^R$ and the following 
	lemma (for a proof see \cite{Lieb-Yngvason:98} or 
	\cite[Lemma 2.5]{Lieb-Seiringer-Solovej-Yngvason:05}).
	\begin{lemma}[Generalized Dyson's Lemma in 3D]\label{lem:dyson_generalized}
		Let $v(r) \geq 0$ be an interaction potential with finite range $R_0$ and scattering 
		length $a$, and $U(r) \geq 0$ be any function satisfying $\int_{0}^{\infty}U(r)r^2\,dr \leq 1$ 
		and $U(r) = 0$ for $r < R_0$. Let $\Omega \subset \R^3$ be star-shaped with 
		respect to $0$. Then for $\psi \in H^1(\Omega)$,
		\begin{align*}
			\int_{\Omega}\left(|\nabla \psi|^2 + \frac{1}{2}v(|\bx|)\psi|^2\right)d\bx 
			\geq a \int_{\Omega}U(|\bx|)|\psi|^2\,d\bx.
		\end{align*}
	\end{lemma}
	
	We will first give an elementary bound for the potential $W_{\beta}^R(\bx)$ on its support, 
	using the relation
	\begin{align}\label{eq:relation_const_w_a_beta}
		\frac{W_{0}}{2} = \left(a_{\beta}/\Lambda_\beta\right)^{\beta-2}.
	\end{align}
	If $\sqrt{3}|Q|^{1/3} \le \delta a_{\beta}$, then
	\begin{align*}
		\frac{1}{2}\frac{W_{0}}{R^{\beta}} = \frac{W_{0}/2}{(\sqrt{3}|Q|^{1/3})^{\beta}} 
		= 3^{-\beta/2}\left(\frac{a_{\beta}}{\Lambda_{\beta}}\right)^{\beta-2}|Q|^{1-\beta/3}|Q|^{-1} 
		\ge \frac{a_{\beta}}{|Q|}\frac{3^{-3/2}}{\Lambda_{\beta}}\left(\delta\Lambda_{\beta}\right)^{3-\beta},
	\end{align*}
	while if $\sqrt{3}|Q|^{1/3} \geq \delta a_{\beta}$, one finds
	\begin{align*}
		\frac{1}{2}\frac{W_{0}}{R^{\beta}} = \frac{W_{0}/2}{(\delta a_{\beta})^{\beta}} 
		= \delta^{-\beta}\left(\frac{a_{\beta}}{\Lambda_{\beta}}\right)^{\beta-2}a_{\beta}^{-\beta} 
		\ge \frac{a_{\beta}}{|Q|}\frac{3^{-3/2}}{\Lambda_{\beta}}\left(\delta\Lambda_{\beta}\right)^{3-\beta}.
	\end{align*}
	This gives
	\begin{multline}\label{eq:estimate_supp_w}
		\int_{Q}\left[|\nabla_2\psi|^2 + \frac{1}{2}W_{\beta}^R(\bx_1-\bx_2)|\psi|^2\right]\,d\bx_1\\
		\ge \frac{a_{\beta}}{|Q|}\frac{3^{-3/2}}{\Lambda_{\beta}}\left(\delta\Lambda_{\beta}\right)^{3-\beta} 
		\int_Q \chi_{B_{\bx_2,R}}(\bx_1) |\psi|^2\,d\bx_1,
	\end{multline}
	and this bound will prove sufficient for the case of small $Q$, where the 
	ball $B_{\bx_2,R}$ covers all of $Q$.
	
	Outside the support of $W_{\beta}^R(\bx)$, namely on $\Omega \cap B_{0,R}^c$ 
	with $\Omega := Q - \bx_2$, 
	we use Lemma~\ref{lem:dyson_generalized} with the potential 
	(note that we may assume $R < \sqrt{3}|Q|^{1/3}$ by the above)
	\begin{align*}
		U(r) = \frac{1}{\sqrt{3}|Q|}\chi_{[R,\sqrt{3}|Q|^{1/3}]}(r),
	\end{align*}
	which satisfies all assumptions of the lemma. This then yields
	\begin{multline}\label{eq:estimate_comp_supp_w}
		\int_{Q}\left[|\nabla_1\psi|^2 + \frac{1}{2}W_{\beta}^R(\bx_1-\bx_2)|\psi|^2\right]\,d\bx_1
		= \int_{\Omega}\left[|\nabla_{\br}\psi|^2 + \frac{1}{2}W_{\beta}^R(\br)|\psi|^2\right]\,d\br\\
		\geq a_{\beta}^R\int_{\Omega}U(|\boldsymbol{r}|) |\psi|^2\,d\boldsymbol{r} 
		= \frac{a_{\beta}^R}{\sqrt{3}|Q|}
		\int_{Q}\chi_{B_{\bx_2,R}^c}(\bx_1)|\psi|^2\,d\bx_1,
	\end{multline}
	where $a_{\beta}^R$, the scattering length of $W_{\beta}^R$, is given by \eqref{eq:scattering_l_w_cut}. 
	We find, using \eqref{eq:relation_const_w_a_beta} and $R=\delta a_{\beta}$, that
	\begin{align*}
		a_{\beta}^R &= a_{\beta}\delta\left(1- \left(\delta\Lambda_{\beta}\right)^{\frac{\beta-2}{2}}\tanh
		\left(\left(\delta\Lambda_{\beta}\right)^{-\frac{\beta-2}{2}}\right)\right).
	\end{align*}
	Combining and integrating the estimates \eqref{eq:estimate_supp_w} 
	and \eqref{eq:estimate_comp_supp_w} yields the bound
	\begin{align*}
		e_2(|Q|;W_{\beta}) 
		\geq e_2(|Q|;W_{\beta}^R) 
		\geq C_{\beta,\delta}\frac{a_{\beta}}{|Q|^{1/3}},
	\end{align*}
	where
	\begin{align}
		C_{\beta,\delta} := \frac{1}{\Lambda_{\beta}}\min\left\{\frac{1}{3^{3/2}}\left(\delta
		\Lambda_{\beta}\right)^{3-\beta},
		\frac{1}{\sqrt{3}}(\delta\Lambda_{\beta})
		\left(1- \left(\delta\Lambda_{\beta}\right)^{\frac{\beta-2}{2}}\tanh
		\left(\left(\delta\Lambda_{\beta}\right)^{-\frac{\beta-2}{2}}\right)\right)\right\}.
	\end{align}
	For simplicity we choose $\delta:= \Lambda_{\beta}^{-1}$, so that $C_{\beta,\delta} = 3^{-1/2}(1-\tanh1)
	\Lambda_{\beta}^{-1}$.
\end{proof}

We then set $\gamma(|Q|) := a_{\beta}|Q|^{-1/3}$ (i.e. $\alpha = 3$, $\calpha = a_{\beta})$ and obtain
\begin{align*}
	e_2(|Q|;W_{\beta}) \geq \zeta \Lambda_{\beta}^{-1}\gamma(|Q|).
\end{align*}
Hence Assumption~\ref{ass:local_exclusion_alpha} holds with 
$e(\gamma) = \zeta \Lambda_{\beta}^{-1}\gamma$ by Theorem~\ref{thm:local_ex_TW},
and we can use 
$\underline{e}_{K=\pi^2}(\gamma) = \min\left\{\zeta \Lambda_{\beta}^{-1}\gamma,\pi^2\right\}$
to obtain the following theorem.

\begin{theorem}\label{thm:hom_scatl_LT}
	Let $d=3$ and $W_{\beta}(\bx) = W_0 |\bx|^{-\beta}$ with $W_0>0$, $\beta > 3$ 
	and scattering length $a_{\beta}$ given
	in \eqref{def:hom_scat_l}. Then for any $N \geq 1$ and all normalized 
	$\psi \in H^1(\R^{3N})$ one has
	\begin{align*}
		\langle\psi, (\hat{T}+\hat{W}_{\beta})\psi\rangle 
		\geq C_{\hs}\int_{\R^3}\min\left\{2^{-1/3}\zeta \Lambda_{\beta}^{-1}a_{\beta}
		\rho(\bx)^2,\pi^2\rho(\bx)^{5/3}\right\}\,d\bx,
	\end{align*}
	where $\rho$ the density associated to $\psi$,
	$\zeta$ is given in Proposition~\ref{prop:hom_pot_e_est}
	and $C_{\hs}$ agrees with the constant found in Theorem~\ref{thm:L-T_HS}.
\end{theorem}
\begin{proof}
	With the above parameters, we apply Theorem~\ref{thm:L-T_global_alpha} with 
	$d=3, \alpha = 3$ and $K=\pi^2$. The statement of the theorem then holds with 
	$C_{3,3,\pi^2} = C_{\hs}$.
\end{proof}
\begin{remark}
	It is instructive to analyze the behavior of $\Lambda_{\beta}^{-1}a_{\beta}$ in the following 
	two limiting cases:
	\begin{enumerate}
		\item $\beta \to \infty$: If we set $W_{0} = a^{\beta-2}$, this should correspond to the 
		hard-sphere case with range $a$. Indeed, in this limit $\Lambda_{\beta} \to 1$ and 
		$a_{\beta} \to a$,
		so we retrieve a bound of the form given above in Theorem~\ref{thm:L-T_HS}.
		
		\item $\beta \to 3$: In this situation, the scattering length $a_{\beta}$ tends to infinity, 
		but $\Lambda_{\beta}^{-1}a_{\beta} = (W_{0}/2)^{1/(\beta-2)} \to W_{0}/2$, 
		so we retrieve the form given by the elementary bound in 
		Theorem~\ref{thm:L-T_global_elementary}. Note that since the scattering length 
		tends to infinity, we cannot have a bound of the form 
		$C\int \min\{ a_\beta\rho^2 ,\rho^{5/3}\}$ for all $\beta > 3$ 
		(with $C$ independent of $\beta$),
		because this would tend to 
		$C \int \rho^{5/3}$ in the limit, which violates scaling 
		(see Proposition~\ref{prop:cl_LT_counterex}).
	\end{enumerate}
\end{remark}

\subsection{Hard-Disk Interaction}
The hard-disk interaction of range $a$ (in $d=2$) corresponds to the potential
\begin{align*}
	W_{a}^{\hd}(\bx)= \left\{ \begin{array}{ll}
 					+\infty ,&\quad |\bx| \leq a,\\
					0, &\quad |\bx|>a.
					\end{array} \right.
\end{align*}
As an operator, $W_{a}^{\hd}$ is realized by requiring that the many-particle wave function $\psi(\bx_1,\dots, \bx_N)$ vanishes for $|\bx_i-\bx_j| \le a$, $1\le i<j \le N$. This subspace of $H^1(\R^{2N})$ is denoted by $H_a^1(\R^{2N})$. 

The two-particle energy $e_2(|Q|;W_{a}^{\hd})$ then becomes
\begin{align*}
	e_2(|Q|;W_a^{\hd}) = |Q|\inf_{\substack{\psi \in H_a^1(\R^{4})\\ \int_{Q^2}|\psi|^2=1}}\,\int_{Q^2}\left[|
	\nabla_1\psi|^2 + |\nabla_2\psi|^2\right]\,d\bx_1d\bx_2.
\end{align*}
We recall that the scattering length of $W_{a}^{\hd}$ is given by the range of the hard disk, namely $a$.
\begin{proposition} \label{prop:hard_disk_e_est}
	In the hard-disk case, the two-particle energy satisfies the estimate
	\begin{align*}
		e_2(|Q|;W_a^{\hd}) \ge \frac{2}{\left(-\ln\left(2^{-1/2}a|Q|^{-1/2}\right)\right)_+}.
	\end{align*}
\end{proposition}
\begin{proof}
	We begin by noting that if $\sqrt{2}|Q|^{1/2} \leq a$, then the estimate is trivial, since 
	the inter-particle distance is less than $a$ and the energy is arbitrarily large. From now on 
	we can thus assume that $\sqrt{2}|Q|^{1/2} > a$. 
	In this situation we will need the following lemma 
	(see \cite{Lieb-Yngvason:01} or \cite[Lemma~3.1]{Lieb-Seiringer-Solovej-Yngvason:05}).
	\begin{lemma}[Generalized Dyson's Lemma in 2D]\label{lem:dyson_generalized_2d}
		Let $v(r)\geq 0$ be an interaction potential with finite range $R_0$ and scattering length $a$, 
		and $U(r) \geq 0$ any function satisfying $\int_0^{\infty}U(r)\ln(r/a)r\,dr \le 1$ and 
		$U(r) = 0$ for $r< R_0$. Let $\Omega$ be star-shaped with respect to $0$. 
		Then any $\psi \in H^1(\Omega)$ satisfies
		\begin{align*}
			\int_{\Omega}\left(|\nabla \psi|^2 + \frac{1}{2}v(|\bx|)|\psi|^2\right)d\bx 
			\geq \int_{\Omega}U(|\bx|)|\psi|^2\,d\bx.
		\end{align*}
	\end{lemma}
	As mentioned before, $W_a^{\hd}$ is implemented through appropriate Dirichlet boundary 
	conditions, so the lemma in this case reads as follows. 
	For $\psi \in H^1(\Omega)$ with $\psi(\bx) = 0$ when $|\bx| \leq a$ (and $R_0 = a$), we have
	\begin{align}\label{eq:dyson_hd}
		\int_{\Omega}|\nabla \psi|^2\,d\bx \geq \int_{\Omega}U(|\bx|)|\psi|^2\,d\bx.
	\end{align}
	To use the above, we define
	\begin{align*}
		U(r) = \frac{1}{|Q|\ln\left(\sqrt{2}|Q|^{1/2}/a\right)}\chi_{[a,\sqrt{2}|Q|^{1/2}]}(r),
	\end{align*}
	so that $U(r) \geq 0$ and
	\begin{align*}
		\int_{0}^{\infty}U(r)\ln(r/a)r\,dr 
		&=  \frac{1}{|Q|\ln\left(\sqrt{2}|Q|^{1/2}/a\right)}\int_{a}^{\sqrt{2}|Q|^{1/2}}\ln(r/a)r\,dr\\
		&\leq \frac{1}{|Q|}\int_{a}^{\sqrt{2}|Q|^{1/2}}r\,dr \leq 1.
	\end{align*}
	Hence for $\bx_2 \in Q$ fixed, \eqref{eq:dyson_hd} yields
	\begin{align*}
		\int_{Q}|\nabla_1\psi|^2\,d\bx_1 
		&\geq \frac{1}{|Q|\ln\left(\sqrt{2}|Q|^{1/2}/a\right)}\int_{Q}\chi_{B_{\bx_2,a}^c}|\psi|^2\,d\bx_1\\
		&= \frac{1}{|Q|\ln\left(\sqrt{2}|Q|^{1/2}/a\right)}\int_{Q}|\psi|^2\,d\bx_1,
	\end{align*}
	since $\psi \equiv 0$ inside $B_{\bx_2,a}$. Repeating the same argument for 
	the term $\int_{Q}|\nabla_2\psi|
	^2\,d\bx_2$ and integrating over the respective variables yields
	\begin{align*}
		e_2(|Q|;W_a^{\hd}) \ge \frac{2}{-\ln\left(2^{-1/2}a|Q|^{-1/2}\right)}.
	\end{align*}
	We then obtain the desired estimate if we set the particle energy to be 
	$+\infty$ if $\sqrt{2}|Q|^{1/2} \leq a$ 
	through considering only the positive part of the denominator.
\end{proof}

If we set $\gamma(|Q|) := a|Q|^{-1/2}$ (i.e. $\alpha = 3$, $\calpha = a$), 
then
\begin{align*}
	e_2(|Q|;W_a^{\hd}) \ge \frac{2}{\left(-\ln\left(2^{-1/2}\gamma\right)\right)_+},
\end{align*}
but this lower bound is \emph{not} concave in $\gamma$. However, it is shown in Appendix~\ref{app:hd_concave_calc} that 
\begin{align*}
	e(\gamma) := \frac{2}{2+\left(-\ln\left(2^{-1/2}\gamma\right)\right)_+}
\end{align*}
is a suitable candidate for a bounded concave function in $\gamma$, and hence 
Assumption~\ref{ass:local_exclusion_alpha} holds with $e = \underline{e}_{K=1}$.
Also note that Assumption~\ref{ass:local_uncert_alpha} holds for this value of $\alpha$
by virtue of Proposition~\ref{prop:uncertainty_2d} in the Appendix.

\begin{theorem}\label{thm:L-T_hd}
	Let $W_{a}^{\hd}$ denote the hard-disk interaction of range $a>0$. 
	There exists a positive constant $C_{\hd}$ independent of $a$ 
	such that for any $N \geq 1$
	and all normalized $\psi \in H_a^1(\R^{2N})$ the total energy satisfies the estimate
	\begin{align}\label{eq:L-T_hd}
		\langle\psi, (\hat{T}+\hat{W}_a^{\hd})\psi\rangle 
		\geq C_{\hd}\int_{\R^2}\frac{2\rho(\bx)^2}{2+\left(-\ln\left(2^{-1}a\rho(\bx)^{1/2}\right)\right)_+}\,d\bx,
	\end{align}
	where $\rho$ is the density associated to $\psi$.
\end{theorem}
\begin{proof}
	With the above parameters, we apply Theorem~\ref{thm:L-T_global_alpha} with 
	$d=2, \alpha = 3$ and $K=1$. The statement of the theorem then holds with 
	$C_{\hd} := C_{2,3,1}$, and taking the remark following 
	Proposition~\ref{prop:uncertainty_2d} into account, we have an explicit lower bound 
	$C_{2,3,1}=c_5/2 > 10^{-10}$.
\end{proof}
\begin{remark}
	It is again instructive to compare the right-hand side
	in \eqref{eq:L-T_hd} for small $\rho(\bx)$ to the exact 
	energy per unit area $e_0$ of the two-dimensional dilute 
	Bose gas, proven in \cite{Lieb-Yngvason:01},
	\begin{align*}
		\lim_{a^2\rho \to 0}\frac{e_0}{4\pi \rho^2|\ln(a^2\rho)|^{-1}} = 1.
	\end{align*}
\end{remark}

\subsection{Estimates for Homogeneous Potentials in Terms of the Scattering Length in 2D}
In this subsection, we will prove a more refined estimate than the
elementary bound (Theorem~\ref{thm:L-T_global_elementary})
for the homogeneous potential 
$W_{\beta}(\bx) = W_{0}|\bx|^{-\beta}$ in 2D in terms of its scattering length. 
Such an estimate is possible provided $\beta > 2$, for which the scattering length is finite.

From \eqref{eq:scattering_l_homogeneous_2d} in the Appendix we obtain 
the scattering length for $W_\beta$:
\begin{align}\label{def:hom_scat_l_2d}
	a_{\beta} = \Xi_{\beta} \left(\frac{W_{0}}{2}\right)^{1/(\beta-2)}, 
	\quad \text{where }\Xi_{\beta} := \left(\frac{2}{\beta-2}\right)^{2/(\beta-2)}.
\end{align}
\begin{proposition}\label{prop:hom_pot_e_est_2d}
	For $W_{\beta}$ with $\beta > 2$, the two-particle energy satisfies
	\begin{align*}
		e_2(|Q|;W_{\beta}) 
		\geq \frac{1}{\zeta_2 + \left(-\ln\left(\frac{a_{\beta}/\Xi_{\beta}}{\sqrt{2}|Q|^{1/2}}\right)\right)_+},
	\end{align*}
	where $\zeta_2 := \mathcal{I}_0(1)/\mathcal{I}_1(1) > 2.24$ is a ratio of Bessel functions.
\end{proposition}
\begin{proof}
	As in the 3D case, our strategy is to replace $W_{\beta}(\bx)$ by a cut-off version
	\begin{align*}
		W_{\beta}^R(\bx) := \frac{W_{0}}{R^{\beta}}\chi_{B_{0,R}}(\bx),
	\end{align*}
	which has a finite range $R$ and a related scattering length 
	$a_{\beta}^R$ (see \eqref{eq:scattering_l_w_cut_2d} 
	in the Appendix), where we choose
	\begin{align}
		R:= \min\{\delta a_{\beta},\sqrt{2}|Q|^{1/2}\}
	\end{align}
	and $\delta > 0$ is to be determined later. 
	
	We will first give an elementary bound for the potential $W_{\beta}^R(\bx)$ on its support, 
	using the relation
	\begin{align}\label{eq:relation_const_w_a_beta_2d}
		\frac{W_{0}}{2} = \left(a_{\beta}/\Xi_\beta\right)^{\beta-2}.
	\end{align}
	If $\sqrt{2}|Q|^{1/2} \le \delta a_{\beta}$, then
	\begin{align*}
		\frac{1}{2}\frac{W_{0}}{R^{\beta}} = \frac{W_{0}/2}{(\sqrt{2}|Q|^{1/2})^{\beta}} 
		= 2^{-\beta/2}\left(\frac{a_{\beta}}{\Xi_{\beta}}\right)^{\beta-2}|Q|^{1-\beta/2}|Q|^{-1} 
		\ge \frac{1}{2|Q|}\left(\delta\Xi_{\beta}\right)^{2-\beta},
	\end{align*}
	while if $\sqrt{2}|Q|^{1/2} \geq \delta a_{\beta}$, one finds
	\begin{align*}
		\frac{1}{2}\frac{W_{0}}{R^{\beta}} = \frac{W_{0}/2}{(\delta a_{\beta})^{\beta}} 
		= \delta^{-\beta}\left(\frac{a_{\beta}}{\Xi_{\beta}}\right)^{\beta-2}a_{\beta}^{-\beta}
		= \frac{\left(\delta \Xi_{\beta}\right)^{2-\beta}}{(\delta a_{\beta})^2}
		\ge \frac{1}{2|Q|}\left(\delta\Xi_{\beta}\right)^{2-\beta}.
	\end{align*}
	This gives
	\begin{multline}\label{eq:estimate_supp_w_2d}
		\int_{Q}\left[|\nabla_2\psi|^2 + \frac{1}{2}W_{\beta}^R(\bx_1-\bx_2)|\psi|^2\right]\,d\bx_1\\
		\ge \frac{1}{2|Q|}\left(\delta\Xi_{\beta}\right)^{2-\beta} 
		\int_Q \chi_{B_{\bx_2,R}}(\bx_1) |\psi|^2\,d\bx_1,
	\end{multline}
	and this bound will prove sufficient for the case of small $Q$, where the 
	ball $B_{\bx_2,R}$ covers all of $Q$.
	
	Outside the support of $W_{\beta}^R(\bx)$, namely on $\Omega \cap B_{0,R}^c$ 
	with $\Omega := Q - \bx_2$, 
	we use Lemma~\ref{lem:dyson_generalized_2d} with the potential 
	(note that we may assume $R < \sqrt{2}|Q|^{1/2}$ by the above)
	\begin{align*}
		U(r) = \frac{1}{|Q|\ln\left(\sqrt{2}|Q|^{1/2}/a_{\beta}^R\right)}\chi_{[R,\sqrt{2}|Q|^{1/2}]}(r),
	\end{align*}
	which satisfies all assumptions of the lemma 
	(it will turn out that $a_{\beta}^R \leq R$), 
	analogously to the hard-disk case with $a=a_{\beta}^R$. 
	This then yields
	\begin{multline}\label{eq:estimate_comp_supp_w_2d}
		\int_{Q}\left[|\nabla_1\psi|^2 + \frac{1}{2}W_{\beta}^R(\bx_1-\bx_2)|\psi|^2\right]\,d\bx_1
		= \int_{\Omega}\left[|\nabla_{\br}\psi|^2 + \frac{1}{2}W_{\beta}^R(\br)|\psi|^2\right]\,d\br\\
		\geq \int_{\Omega}U(|\boldsymbol{r}|) |\psi|^2\,d\boldsymbol{r} 
		= \frac{1}{|Q|\ln\left(\sqrt{2}|Q|^{1/2}/a_{\beta}^R\right)}
		\int_{Q}\chi_{B_{\bx_2,R}^c}(\bx_1)|\psi|^2\,d\bx_1,
	\end{multline}
	where $a_{\beta}^R$, the scattering length of $W_{\beta}^R$, is given by 
	\eqref{eq:scattering_l_w_cut_2d}. 
	We find, using \eqref{eq:relation_const_w_a_beta_2d} and $R=\delta a_{\beta}$, that
	\begin{align*}
		a_{\beta}^R 
		= a_{\beta}\delta\exp\left(-(\delta\Xi_{\beta})^{\frac{\beta-2}{2}}
		\frac{\mathcal{I}_0\left((\delta\Xi_{\beta})^{-\frac{\beta-2}{2}}\right)}
		{\mathcal{I}_1\left((\delta\Xi_{\beta})^{-\frac{\beta-2}{2}}\right)}\right).
	\end{align*}
	Combining and integrating the estimates \eqref{eq:estimate_supp_w_2d} 
	and \eqref{eq:estimate_comp_supp_w_2d}, and choosing $\delta:= \Xi_{\beta}^{-1}$ 
	for simplicity, yields the bound
	\begin{align*}
		e_2(|Q|;W_{\beta}) &\geq \min\left\{\frac{1}{2}, 
		\left[ \ln\left(\frac{\sqrt{2}|Q|^{1/2}}{a_{\beta}/\Xi_\beta}
		\exp\left(\frac{\mathcal{I}_0(1)}{\mathcal{I}_1(1)}\right)\right)\right]^{-1}\right\}\\
		&= \frac{1}{\zeta_2 + \left(-\ln\left(\frac{a_{\beta}/\Xi_{\beta}}{\sqrt{2}|Q|^{1/2}}\right)\right)_+},
	\end{align*}
	where $\zeta_2 := \mathcal{I}_0(1)/\mathcal{I}_1(1) > 2.24$.
\end{proof}

We then set $\gamma(|Q|) := a_{\beta}|Q|^{-1/2}$ (i.e. $\alpha = 3$, $\calpha = a_{\beta})$ and obtain
\begin{align*}
	e_2(|Q|;W_{\beta}) \geq e(\gamma) 
	:= \frac{1}{\zeta_2 + \left(-\ln\left(2^{-1/2}\gamma/\Xi_{\beta}\right)\right)_+}
	= \underline{e}_{K=1}(\gamma).
\end{align*}
Hence, following the hard-disk case, we obtain the following theorem.

\begin{theorem}\label{cor:hom_scatl_LT_2d}
	Let $d=2$ and $W_{\beta}(\bx) = W_0 |\bx|^{-\beta}$ with $W_0>0$, $\beta > 2$ and 
	scattering length $a_{\beta}$ given
	in \eqref{def:hom_scat_l_2d}. Then for any $N \geq 1$ and all 
	normalized $\psi \in H^1(\R^{2N})$ one has
	\begin{align*}
		\langle\psi, (\hat{T}+\hat{W}_{\beta})\psi\rangle 
		\geq C_{\hd}\int_{\R^2}\frac{\rho(\bx)^2}
		{\zeta_2+\left(-\ln\left(2^{-1}\Xi_{\beta}^{-1}a_{\beta}\rho(\bx)^{1/2}\right)\right)_+}\,d\bx,
	\end{align*}
	where $\rho$ is the density associated to $\psi$,
	$\zeta_2$ is given in Proposition~\ref{prop:hom_pot_e_est_2d} and 
	$C_{\hd}$ is as in Theorem~\ref{thm:L-T_hd}.
\end{theorem}
\begin{remark}
	It is instructive to analyze the behavior of $\Xi_{\beta}^{-1}a_{\beta}$ in the following 
	two limiting cases:
	\begin{enumerate}
		\item $\beta \to \infty$: If we set $W_{0} = a^{\beta-2}$, this should correspond to the 
		hard-disk case with range $a$. Indeed, in this limit $\Xi_{\beta} \to 1$ and 
		$a_{\beta} \to a$,
		so we retrieve a bound of the form given above in Theorem~\ref{thm:L-T_hd}.
		
		\item $\beta \to 2$: In this situation, the scattering length $a_{\beta}$ tends to infinity, 
		however $\Xi_{\beta}^{-1}a_{\beta} = (W_{0}/2)^{1/(\beta-2)}$
		could have a finite limit depending on the size of the coupling constant.
		If $W_{0}$ is sufficiently large, we obtain
		the classical Lieb-Thirring estimate (see Theorem~\ref{thm:square_int_L-T}) with a constant
		$C_{\hd}/\zeta_2$. A more natural dependence of the Lieb-Thirring constant
		on $W_{0}$ could possibly be recovered by a different choice
		of $\delta$.
	\end{enumerate}
\end{remark}

\section{Counterexamples}\label{sec:counter_ex}
In this section, we investigate the sharpness of the forms of the previously obtained 
Lieb-Thirring type bounds. We restrict the discussion to $d=3$ for simplicity.

\subsection{Homogeneous Potentials}
As seen in Theorem~\ref{thm:square_int_L-T}, the classical Lieb-Thirring estimate can be recovered
from the class of inverse-square interaction potentials $W_2(\bx) = W_{0}|\bx|^{-2}$.
A natural question would be to ask whether this is also possible for other 
homogeneous potentials $W_{\beta}$. 
The following proposition shows that this is
impossible for $\beta \neq 2$ by using appropriate scaling and a suitable trial function.
\begin{proposition}\label{prop:cl_LT_counterex}
	Let $W_{\beta}(\bx) = W_{0}|\bx|^{-\beta}$ with $\beta \neq 2$. Assume that
	there exists a constant
	$C\geq 0$ such that the inequality
	\begin{align}\label{eq:cl_LT_counterex}
		\langle\psi, (\hat{T}+\hat{W}_{\beta})\psi\rangle 
		\geq C\int_{\R^3}\rho(\bx)^{5/3}\,d\bx
	\end{align}
	holds for all $\psi \in H^1(\R^{3N})$ and $N \geq 1$ 
	($\rho$ being the density of $\psi$). Then $C = 0$.
\end{proposition}
\begin{remark}
	Note the implication of this proposition in the remark following 
	Theorem~\ref{thm:hom_scatl_LT}.
\end{remark}
\begin{proof}
	Given $\psi \in L^2(\R^{3N})$, define by 
	\begin{align*}
		\psi_L(\bx_1,\dots,\bx_N) := L^{-3N/2}\psi(\bx_1/L,\dots,\bx_N/L)
	\end{align*} 
	a scaled version of $\psi$. We have the density 
	$\rho_{\psi_L}(\bx) = L^{-3}\rho_{\psi}(\bx/L)$, so that
	\begin{align*}
		\int_{\R^3}\rho_{\psi_L}^{5/3}(\bx)\,d\bx = L^{-2}\int_{\R^3}\rho_{\psi}^{5/3}(\by)\,d\by,
	\end{align*}
	and furthermore
	\begin{align*}
		\langle\psi_L,(\hat{T}+\hat{W}_{\beta})\psi_L\rangle 
		= L^{-2}\langle\psi,\hat{T}\psi\rangle + L^{-\beta}\langle\psi,\hat{W}_{\beta}\psi\rangle.
	\end{align*}
	The scaled version of \eqref{eq:cl_LT_counterex} then reads
	\begin{align*}
		\langle\psi,\hat{T}\psi\rangle + L^{2-\beta}\langle\psi,\hat{W}_{\beta}\psi\rangle
		\geq C\int_{\R^3}\rho_{\psi}^{5/3}(\by)\,d\by.
	\end{align*}
	
	Set $\psi = \psi_{\hs} \in H_a^1(\R^{3N})$, the ground state of 
	the hard-sphere potential with range
	$a$ on the cube $Q_0$ with Dirichlet boundary conditions, 
	and define $\bar{\rho} = N/|Q_0|$. We decide to choose $Q_0$ (or $a$) 
	such that the state is dilute, meaning
	$\bar{\rho}a^{3} < \eps \ll 1$. For this state it holds, for sufficiently large $N$ and $Q_0$, 
	that (see \cite[Theorem~3]{Dyson:57} or \cite[Theorem~2.1]{Lieb-Yngvason:00})
	\begin{align*}
		\langle\psi,\hat{T}\psi\rangle &\leq c_1a \frac{N^2}{|Q_0|},
	\end{align*}
	for a sufficiently large $c_1$, and furthermore 
	$\langle\psi,\hat{W}_{\beta}\psi\rangle \leq c_2 \frac{N^2}{a^{\beta}}$,
	so we obtain
	\begin{align*}
		\langle\psi,\hat{T}\psi\rangle 
		+ L^{2-\beta}\langle\psi,\hat{W}_{\beta}\psi\rangle
		\leq c_1 a \frac{N^2}{|Q_0|} + c_2a^{-\beta}\frac{N^2}{L^{\beta-2}}.
	\end{align*}
	We then estimate
	\begin{align*}
		aN^2/|Q_0| = N\bar{\rho}^{2/3}(a^3\bar{\rho})^{1/3} 
		< \eps^{1/3}N\bar{\rho}^{2/3} = \eps^{1/3}\frac{N^{5/3}}{|Q_0|^{2/3}}.
	\end{align*}
	Applying H\"older's inequality to $N = \int_{Q_0}\rho_{\psi}$ yields
	\begin{align*}
		\frac{N^{5/3}}{|Q_0|^{2/3}} \leq \int_{\R^3} \rho_{\psi}^{5/3}(\bx)\,d\bx,	
	\end{align*}
	so that (for positive $C$)
	\begin{align*}
		c_1 a \frac{N^2}{|Q_0|} + c_2a^{-\beta}\frac{N^2}{L^{\beta-2}} 
		&\leq \eps^{1/3}c_1\int_{\R^3}\rho_{\psi}^{5/3}(\bx)\,d\bx 
		+ c_2a^{-\beta}\frac{N^2}{L^{\beta-2}}\\
		&< C\int_{\R^3}\rho_{\psi}^{5/3}(\bx)\,d\bx
	\end{align*}
	for $\eps$ small enough and $L \to 0$ for $\beta < 2$ or $L \to \infty$ for $\beta > 2$. Hence
	\eqref{eq:cl_LT_counterex} is impossible if $C$ is independent of $N$ (and $\psi$).
\end{proof}

\subsection{Locally Integrable Potentials}
As seen in Theorem~\ref{thm:hom_scatl_LT}, the homogeneous potentials $W_{\beta}$
with $\beta > 3$ satisfy an estimate involving the scattering length $a_{\beta}$. 
When $\beta \leq 3$, the scattering length of these potentials becomes infinite, reducing
a possible estimate of this type to the classical Lieb-Thirring estimate. As shown in the previous
subsection, this is only possible if $\beta = 2$. One is then tempted to ask if
an estimate as in Theorem~\ref{thm:hom_scatl_LT} could hold for some other class of
potentials. The following shows that if the potential is sufficiently regular,
then such a bound cannot hold.
\begin{proposition}
	Let $W \in L_{\loc}^{p}(\R^3)$ for some $p\geq 3/2$ and scattering 
	length $a>0$ (possibly infinite). 
	If there exists a constant $C \geq 0$ such that the inequality
	\begin{align}\label{eq:a_LT_counterex_bounded}
		\langle\psi, (\hat{T}+\hat{W})\psi\rangle 
		\geq C\int_{\R^3}\min\{a\rho(\bx)^2,\rho(\bx)^{5/3}\}\,d\bx
	\end{align}
	holds for all $\psi \in H^1(\R^{3N})$ and $N \geq 1$, then $C=0$.
\end{proposition}
\begin{remark}
	Observe that $W_2 \in L_{\loc}^{p}(\R^3)$ for $p<3/2$
	and has infinite scattering length, and
	we have seen that \eqref{eq:a_LT_counterex_bounded} holds for
	this class of potentials. Also, note that
	$W_{\beta}$ has finite scattering length for $\beta > 3$ and 
	\eqref{eq:a_LT_counterex_bounded} holds for this potential, 
	but it is not in $L_{\loc}^p(\R^3)$ for any $p \geq 1$.
\end{remark}
\begin{proof}
	Let $\vphi(\bx) \in C_c^{\infty}(\R^3;\R_{\ge 0})$ with $\supp \vphi \subset B_{0,1}$ and 
	$\int_{B_{0,1}}|\vphi(\bx)|^2\,d\bx = 1$. Set $\vphi_{L}(\bx) = L^{-3/2}\vphi(\bx/L)$ and define
	\begin{align*}
		\psi(\bx_1,\dots,\bx_N) := \prod_{j=1}^N\vphi_L(\bx_j).
	\end{align*}
	For this function, the density is 
	$\rho_{\psi}(\bx) = L^{-3}N|\vphi(\bx/L)|^2$, so that
	\begin{multline*}
		\int_{\R^3}\min\{a\rho_{\psi}^2(\bx),\rho_{\psi}(\bx)^{5/3}\}\,d\bx\\
		= \int_{\R^3}\min\left\{a\frac{N^2}{L^3}|\vphi(\by)|^4,
		\frac{N^{5/3}}{L^2}|\vphi(\by)|^{10/3}\right\}\,d\by.
	\end{multline*}
	The kinetic energy on the other hand computes to
	\begin{align*}
		\langle\psi,\hat{T}\psi\rangle 
		= \frac{N}{L^2}\int_{B_{0,1}}|\nabla\vphi(\by)|^2\,d\by \leq c_3 \frac{N}{L^2},
	\end{align*}
	whereas
	\begin{align*}
		\langle\psi, \hat{W}\psi\rangle 
		&= \frac{N(N-1)}{2}\int_{\R^6}|\vphi_L(\bx)|^2W(\bx-\by)|\vphi_L(\by)|^2\,d\bx d\by\\
		&\leq N^2\int_{\R^6}|\vphi(\bx)|^2W(L(\bx-\by))|\vphi(\by)|^2\,d\bx d\by\\
		&\leq c_4N^2\int_{\R^6}\chi_{B_{0,1}}(\bx)W(L(\bx-\by))\chi_{B_{0,1}}(\by)\,d\bx d\by\\
		&\leq c_4N^2\int_{B_{0,1}}d\bx\int_{B_{0,2}}W(L\bz)\,d\bz.
	\end{align*}
	We find
	\begin{align*}
		\int_{B_{0,2}}W(L\bz)\,d\bz = L^{-3}\int_{B_{0,2L}}W(\bz)\,d\bz 
		\leq c_5 L^{-3/p}\norm{W}_{L^p(B_{0,2L})}
	\end{align*}
	by H\"older's inequality. Hence \eqref{eq:a_LT_counterex_bounded}
	is violated if
	\begin{multline}\label{eq:a_counterex_integrable}
		c_3\frac{N}{L^2} + c_4c_5 \frac{N^2}{L^{3/p}}\norm{W}_{L^p(B_{0,2L})}\\ 
		< C \int_{\R^3}\min\left\{a\frac{N^2}{L^3}|\vphi(\by)|^4,
		\frac{N^{5/3}}{L^2}|\vphi(\by)|^{10/3}\right\}\,d\by.
	\end{multline}
	
	Since $L_{\loc}^{p}(\R^3) \subset L_{\loc}^{3/2}(\R^3)$ for $p > 3/2$, it suffices
	to study $p=3/2$. We then rewrite \eqref{eq:a_counterex_integrable} as
	\begin{multline*}
		c_3N + c_4c_5 N^2\norm{W}_{L^{3/2}(B_{0,2L})}\\ 
		< C \int_{\R^3}\min\left\{a\frac{N^2}{L}|\vphi(\by)|^4,
		N^{5/3}|\vphi(\by)|^{10/3}\right\}\,d\by.
	\end{multline*}
	Using $\lim_{L\to 0}\norm{W}_{L^{3/2}(B_{0,2L})} = 0$, we can choose $L= L(N)$
	such that $\norm{W}_{L^{3/2}(B_{0,2L})} < c_6 N^{-1/3}$, so that the left hand side
	is dominated by the right hand side for sufficiently small $c_6$ and $N$ large.
\end{proof}

\subsection{Skew Potentials}
In this subsection we will show that a bound as in Theorem~\ref{thm:hom_scatl_LT} is
also not possible for potentials $W$ with an unbalanced relationship between
the range and the scattering length.
This shows that such a bound cannot depend on the scattering length alone
but must also depend on other details of the potential.

For $0<a<R$ and $W_0>0$ let us define
\begin{align*}
	W_{a,R}(\bx) = \left\{ \begin{array}{ll}
 					+\infty, &\quad |\bx| \leq a,\\
					W_0, &\quad a<|\bx|\leq R,\\
					0, &\quad |\bx|>R,
					\end{array} \right.
\end{align*}
a modified hard-sphere interaction. 
Note that $W_{a,R} \notin L^p_{\loc}(\R^3)$ for any $p$, however,
it is easy to show (cf. Appendix~\ref{app:w_cut_calc}) 
that the scattering length of this class of potentials is given by
\begin{align*}
	a_W := a + \frac{1}{\sqrt{W_0/2}}\left(\sqrt{W_0/2}(R-a) - \tanh\left(\sqrt{W_0/2}(R-a)\right)\right),
\end{align*}
so that $a_W>a$ for $R>a$ and $\lim_{R \to \infty}a_W = \infty$.
\begin{proposition}
	Let $W=W_{a,R}$. Assume that for some constant 
	$C \geq 0$ the inequality
	\begin{align}\label{eq:skew_LT_counterex}
		\langle\psi, (\hat{T}+\hat{W}_{a,R})\psi\rangle 
		\geq C\int_{\R^3}\min\{a_W\rho(\bx)^2,\rho(\bx)^{5/3}\}\,d\bx
	\end{align}
	holds for all $\psi \in H^1(\R^{3N})$ and $N \geq 1$.
	We allow $C$ to depend on $a_W>0$ but not on the details of $W$
	(i.e. $a$, $W_0$, and $R$).
	Then $C = 0$.
\end{proposition}
\begin{proof}
	Fix $a_W > 0$ and assume that $C>0$.
	As in the proof of Proposition~\ref{prop:cl_LT_counterex}, 
	we take $\psi = \psi_{\hs} \in H_a^1(\R^{3N})$, the ground state of 
	the hard-sphere potential with range $a$ on a cube $Q_0$ 
	with Dirichlet boundary conditions, 
	and define $\bar{\rho} = N/|Q_0|$. 
	We start by fixing $0<a \ll a_W$ such that $a < Ca_W/(16\pi)$
	and let $N$ and $Q_0$ be such that
	$(a^3\bar{\rho})^{1/3} < \frac{5}{6} C/(16\pi)$
	and (see \cite[Theorem~3]{Dyson:57} or \cite[Theorem~2.1]{Lieb-Yngvason:00})
	\begin{align*}
		\langle\psi,\hat{T}\psi\rangle &\leq 8\pi a \frac{N^2}{|Q_0|}.
	\end{align*}
	We also note that for this state
	$\langle\psi,\hat{W}_{a,R}\psi\rangle \leq N^2 W_0$. 
	We now estimate
	\begin{align*}
		a\frac{N^2}{|Q_0|} = (a^3\bar{\rho})^{1/3} N\bar{\rho}^{2/3} 
		= (a^3\bar{\rho})^{1/3} \frac{N^{5/3}}{|Q_0|^{2/3}}
		\le (a^3\bar{\rho})^{1/3} \int_{Q_0} \rho_{\psi}^{5/3},
	\end{align*}
	where we applied H\"older's inequality to $N = \int_{Q_0}\rho_{\psi}$.
	We then write $\int_{Q_0} \rho_{\psi}^{5/3} =$
	\begin{align*}
		\int_{\rho_{\psi} \le a_W^{-3}} \rho_{\psi}^{5/3} 
			+ \int_{\rho_{\psi}>a_W^{-3}} \rho_{\psi}^{5/3}
		\le \left( \int_{\rho_{\psi} \le a_W^{-3}} \rho_{\psi}^2 \right)^{5/6}|Q_0|^{1/6} 
			+ \int_{\rho_{\psi}>a_W^{-3}} \rho_{\psi}^{5/3},
	\end{align*}
	again using H\"older's inequality, 
	and $\int_{\rho_{\psi} \le a_W^{-3}} 1 \le |Q_0|$.
	Furthermore, by means of Young's inequality $ab \le a^p/p + b^q/q$,
	we have
	\begin{align*}
		(a^3\bar{\rho})^{1/3} \left( \int_{\rho_{\psi} \le a_W^{-3}} \rho_{\psi}^2 \right)^{5/6}|Q_0|^{1/6}
		\le \frac{5}{6} a\int_{\rho_{\psi} \le a_W^{-3}} \rho_{\psi}^2
			+ \frac{1}{6} a \bar{\rho}^2 |Q_0|,
	\end{align*}
	and hence
	\begin{align*}
		a\frac{N^2}{|Q_0|} 
		\le \frac{5}{6} a\int_{\rho_{\psi} \le a_W^{-3}} \rho_{\psi}^2
			+ (a^3\bar{\rho})^{1/3} \int_{\rho_{\psi}>a_W^{-3}} \rho_{\psi}^{5/3}
			+ \frac{1}{6} a \frac{N^2}{|Q_0|},
	\end{align*}
	or equivalently,
	$a\frac{N^2}{|Q_0|} 
		\le a\int_{\rho_{\psi} \le a_W^{-3}} \rho_{\psi}^2
			+ \frac{6}{5} (a^3\bar{\rho})^{1/3} 
				\int_{\rho_{\psi}>a_W^{-3}} \rho_{\psi}^{5/3}$.
	Summing up,
	\begin{align*}
		\langle\psi, (\hat{T}+\hat{W}_{a,R})\psi\rangle 
		&\le 8\pi a \frac{N^2}{|Q_0|} + N^2W_0 \\
		&\le 8\pi a\int_{\rho_{\psi} \le a_W^{-3}} \rho_{\psi}^2
			+ 8\pi \frac{6}{5} (a^3\bar{\rho})^{1/3} \int_{\rho_{\psi}>a_W^{-3}} \rho_{\psi}^{5/3}
			+ N^2 W_0 \\
		&< \frac{C}{2} a_W \int_{\rho_{\psi} \le a_W^{-3}} \rho_{\psi}^2
			+ \frac{C}{2} \int_{\rho_{\psi}>a_W^{-3}} \rho_{\psi}^{5/3}
			+ N^2 W_0 \\
		&= \frac{C}{2} \int_{\R^3}\min\{a_W\rho_{\psi}^2,\rho_{\psi}^{5/3}\}
			+ N^2 W_0
	\end{align*}
	Now, take $W_0$ small but $R$ large, keeping $a_W$ fixed,
	so that $N^2 W_0 < \frac{C}{2} \int_{\R^3}\min\{a_W\rho_{\psi}^2,\rho_{\psi}^{5/3}\}$
	(which is strictly positive).
	It follows that \eqref{eq:skew_LT_counterex} is impossible if $C>0$ 
	is independent of $a$, $W_0$, $R$ (and $\psi$).
\end{proof}

\appendix
\section{Appendix}
\subsection{Local Uncertainty Principles}
For our applications with $\alpha > 2$ we require a
stronger formulation of the local uncertainty principle, cf. 
Assumption~\ref{ass:local_uncert_alpha}.

\begin{proposition}[Uncertainty in $d \geq 3$]\label{prop:uncertainty_3geqd}
	Let $d \geq 3$ and $W \geq 0$, then Assumption~\ref{ass:local_uncert_alpha} 
	holds for all $0<\alpha\leq \frac{2d}{d-2}$, with
	$S_1 = S_d/2$ and $S_2=S_d$, where $S_d$
	is the inverse-square of the Poincar\'e-Sobolev constant for the cube in dimension $d$.
\end{proposition}
\begin{remark}
	An upper bound on $\alpha$ in $d \geq 3$ is to be expected, since the kinetic energy
	cannot control arbitrarily large regularity of the density.
\end{remark}
\begin{remark}
	Following \cite[Exercise~8.6]{Lieb-Loss:01}, an explicit lower bound for $S_d$ is
	\begin{align*}
		S_d \geq \frac{16d^2}{d^{d}(d+2)^2}\left(\frac{d}{|\mathbb{S}^{d-1}|}\right)^{\frac{2(d-1)}{d}}
		\left[\left(\frac{2(d-1)}{d}\right)^{\frac{d-1}{d}}+\left(\frac{2(d-1)}{d-2}\right)^{\frac{d-1}{d}}\right]^{-2},
	\end{align*}
	in particular $S_3 \geq 0.00226$. Comparing with the sharp global Sobolev constant
	indicates that this leaves plenty of room for improvement.
\end{remark}
\begin{proof}
	We use that (cf. also \cite{Frank-Seiringer:12})
	\begin{align*}
		(T+W)_{\psi}^{Q} \geq T_{\psi}^{Q} \geq \int_{Q}|\nabla \rho(\bx)^{1/2}|^2\,d\bx,
	\end{align*}
	where the second inequality follows from the Hoffmann-Ostenhof inequality, \cite[Lemma~2]{H-O2:77}. 
	We then apply the Poincar\'e-Sobolev inequality \cite[Theorem~8.12]{Lieb-Loss:01} 
	and the triangle inequality in $L^{2^*}(Q)$, where $2^* := 2d/(d-2)$,
	to obtain
	\begin{align*}
		 \int_{Q}|\nabla \rho(\bx)^{1/2}|^2\,d\bx 
		 &\geq S_d\Norm{\rho^{1/2}-\frac{1}{|Q|}\int_Q\rho^{1/2}}_{L^{2^*}(Q)}^2\\
		 &\geq S_d\left(\norm{\rho^{1/2}}_{L^{2^*}(Q)}-|Q|^{-(d+2)/2d}\int_Q\rho^{1/2}\right)^2.
	\end{align*}
	Next, we use the inequality $(a-b)^2 \geq a^2/2 - b^2$ (convexity) and the Cauchy-Schwarz inequality, 
	so that the right hand side is bounded from below by
	\begin{align*}
		S_d\left(\frac{1}{2}\norm{\rho}_{L^{d/(d-2)}(Q)}-|Q|^{-2/d}\int_Q\rho\right)
		\geq \frac{S_d}{2}\frac{\left(\int_Q\rho^{1+\alpha/d}\right)^{2/\alpha}}{(\int_Q\rho)^{2/d+2/\alpha-1}}
		-S_d\frac{\int_Q\rho}{|Q|^{2/d}},
	\end{align*}
	where for the last step we used the H\"older inequality
	\begin{align*}
		\textstyle{\int_Q\rho^{1+\alpha/d} 
		\leq (\int_Q \rho^{d/(d-2)})^{\alpha(d-2)/(2d)}(\int_Q\rho)^{1-\alpha(d-2)/(2d)}},
	\end{align*}
	which is applicable for $0<\alpha \leq 2^*$.
\end{proof}

\begin{proposition}[Uncertainty in $d = 2$]\label{prop:uncertainty_2d}
	Let $d = 2$ and $W \geq 0$. Then Assumption~\ref{ass:local_uncert_alpha} 
	holds for all $0<\alpha<\infty$, with
	$S_1$ a universal constant depending on $\alpha$, and $S_2 = 1$.
\end{proposition}
\begin{remark}
	Following \cite[Theorem~5.8]{Adams-Fournier:03}, an explicit lower bound for $S_1$ is
	\begin{align*}
		S_1\geq \frac{\pi}{96\sqrt{2}}\left(3\sqrt{2}(4+\alpha)\right)^{-(4+\alpha)/\alpha}.
	\end{align*}
\end{remark}
\begin{proof}
	We again obtain
	\begin{align*}
		(T+W)_{\psi}^{Q} \geq T_{\psi}^{Q} \geq \int_{Q}|\nabla \rho(\bx)^{1/2}|^2\,d\bx
	\end{align*}
	by the Hoffmann-Ostenhof inequality. For the next step, we will need the 
	Gagliardo-Nirenberg-Sobolev interpolation inequality 
	(see \cite[Theorem~5.8]{Adams-Fournier:03}):
	\begin{align*}
		\norm{f}_{L^2(\Omega)} + \norm{\nabla f}_{L^2(\Omega)} &= \norm{f}_{H^1(\Omega)}\\
		&\geq S_{2,p}\norm{f}_{L^2(\Omega)}^{-\frac{2}{p-2}}\norm{f}_{L^p(\Omega)}^{\frac{p}{p-2}}, 
		\quad 2 < p < \infty.
	\end{align*}
	We then apply this inequality taking $\Omega$ to be the unit square, $f=\rho^{1/2}$, $p=2+\alpha$ 
	and use the inequality $(a-b)^2 \geq a^2/2 - b^2$
	and scaling to obtain the following estimate, 
	valid on any square $Q$,
	\begin{align*}
		\int_{Q}|\nabla \rho(\bx)^{1/2}|^2\,d\bx 
		&\geq \frac{S_{2,2+\alpha}^2}{2}
		\frac{\norm{\rho^{1/2}}_{L^{2+\alpha}(Q)}^{2(2+\alpha)/
		\alpha}}
		{\norm{\rho^{1/2}}_{L^2(Q)}^{4/\alpha}}
		- \frac{1}{|Q|}\int_{Q}\rho\\
		&= \frac{S_{2,2+\alpha}^2}{2}\frac{(\int_Q \rho^{1+\alpha/2})^{2/\alpha}}
		{(\int_Q \rho)^{2/\alpha}}
		- \frac{1}{|Q|}\int_{Q}\rho.
	\end{align*}
\end{proof}

\begin{proposition}[Uncertainty in $d = 1$]\label{prop:uncertainty_1d}
	Let $d = 1$ and $W \geq 0$. Then Assumption~\ref{ass:local_uncert_alpha} 
	holds for all $0<\alpha<\infty$, with
	$S_1$ a universal constant depending on $\alpha$, and $S_2 = 1$.
\end{proposition}
\begin{proof}
	The proof is almost identical to the two-dimensional case. We omit the details.
\end{proof}

\subsection{Scattering Lengths}
The reader is referred to the Appendix~C in \cite{Lieb-Seiringer-Solovej-Yngvason:05} for a basic introduction to the concept of scattering lengths and some useful properties.
\subsubsection{Scattering Length for $W_{\beta}$ in 3D}
In this subsection we compute the scattering length for the potential
\begin{align*}
	W_\beta(\bx) = W_{0} |\bx|^{-\beta}, \quad \beta > 3. 
\end{align*}
Hence we look for a solution $\vphi(\bx)$ to the equation
\begin{align}\label{eq:scattering_w_alpha}
	\left(-\Delta + \frac{1}{2}\frac{W_{0}}{|\bx|^{\beta}}\right)\vphi(\bx) = 0
\end{align}
with the asymptotics $\vphi(\bx) = 1-\frac{a_{\beta}}{|\bx|} + O(|\bx|^{-2})$ as $|\bx| \to \infty$, 
where $a_{\beta}$ defines the scattering length in $d=3$ for this potential. 
Through scaling we see that if $\psi(\bx)$ is a solution to the equation
\begin{align}\label{eq:scattering_w_alpha_simple}
	(-\Delta + |\bx|^{-\beta})\psi(\bx) = 0,
\end{align}
then $\vphi(\bx) = \psi((W_{0}/2)^{-1/(\beta-2)}\bx)$ solves \eqref{eq:scattering_w_alpha}, so we first look for solutions to \eqref{eq:scattering_w_alpha_simple} with $\psi(\bx) = 1- \frac{a}{|\bx|} + O(|\bx|^{-2})$.

Radial symmetry and the substitution $u(r) = r \psi(r)$ shows that \eqref{eq:scattering_w_alpha_simple} is equivalent to the differential equation
\begin{align*}
	-u''(r) + r^{-\beta}u(r) = 0, \quad u(0) = 0,
\end{align*}
$u(r)$ having asymptotics $u(r) = r - a + \mathcal{O}(r^{-1})$ for large $r$. We make the Ansatz
\begin{align*}
	u(r) = \sqrt{r}F(2r^{-(\beta-2)/2}/(\beta-2)),
\end{align*}
with $F$ a function to be determined. Let us for convenience define 
$t := 2r^{-(\beta-2)/2}/(\beta-2) = 2\nu r^{-1/(2\nu)}$, with $\nu := 1/(\beta-2) \in (0,1)$. 
Explicit computation shows that $F(t)$ satisfies the modified Bessel equation
\begin{align*}
	t^2 F''(t) + tF'(t) - (t^2 + \nu^2)F(t) = 0,
\end{align*}
so $F(t) = c_1\mathcal{I}_{\nu}(t) + c_2K_{\nu}(t)$, where $\mathcal{I}_{\nu}, K_{\nu}$ are the modified Bessel functions
\begin{align*}
	\mathcal{I}_{\nu}(t) 
	&= \left(\frac{t}{2}\right)^{\nu}\sum_{k=0}^{\infty}\frac{(t^2/4)^{k}}{k! \Gamma(\nu + k +1)},\\
	K_{\nu}(t) &= \frac{\pi}{2}\frac{\mathcal{I}_{-\nu}(t) - \mathcal{I}_{\nu}(t)}{\sin(\pi \nu)}, \quad \nu \notin \mathbb{Z},
\end{align*}
with asymptotics for large $t$
\begin{align*}
	\mathcal{I}_{\nu}(t) &= \frac{e^{t}}{\sqrt{2\pi t}}\left(1 + \mathcal{O}(t^{-1})\right),\\
	K_{\nu}(t) &= \sqrt{\frac{\pi}{2t}}e^{-t}\left(1 + \mathcal{O}(t^{-1})\right),
\end{align*}
see \cite[Chapter 9.6 \& 9.7]{Abramowitz-Stegun:70}. The requirement $u(0)=0$ reduces to the condition
\begin{align*}
	\lim_{t \to \infty}(t/(2\nu))^{-\nu}F(t) = 0,
\end{align*}
hence $c_1 = 0$ and $F(t) = c_2K_{\nu}(t)$. For small $t$, $t^{-\nu}K_{\nu}(t)$ has the expansion
\begin{align*}
	t^{-\nu}K_{\nu}(t) = C\left(\frac{t^{-2\nu}}{\Gamma(1-\nu)} - \frac{1}{\Gamma(\nu+1)} + \mathcal{O}(t)\right),
\end{align*}
hence for large $r$ one has
\begin{align*}
	\sqrt{r}K_{\nu}(2\nu r^{-1/(2\nu)}) 
	&= C\left(\frac{(2\nu)^{-2\nu}}{\Gamma(1-\nu)}r - \frac{1}{\Gamma(\nu
	+1)}+\mathcal{O}(r^{-1})\right)\\
	&= C\left(r - \frac{\Gamma(1-\nu)(2\nu)^{2\nu}}{\Gamma(\nu+1)} 
	+ \mathcal{O}(r^{-1})\right).
\end{align*}
From this we can conclude that the solution (with the right normalization) to \eqref{eq:scattering_w_alpha_simple} has
\begin{align*}
	a = \frac{\Gamma(1-\nu)(2\nu)^{2\nu}}{\Gamma(\nu+1)},
\end{align*}
so expanding $\vphi(\bx) = \psi((W_{0}/2)^{-1/(\beta-2)}\bx)$ gives
\begin{align*}
	\vphi(\bx) = C\left(1-\frac{a}{(W_{0}/2)^{-1/(\beta-2)}}\frac{1}{|\bx|} + O(|\bx|^{-2})\right).
\end{align*}
Thus $W_{\beta}$ has scattering length
\begin{align}\label{eq:scattering_l_homogeneous}
	a_{\beta} = \frac{\Gamma\left(\frac{\beta-3}{\beta-2}\right)}{\Gamma\left(\frac{\beta-1}{\beta-2}\right)}\left(\frac{2}{\beta-2}\right)^{2/(\beta-2)}\left(\frac{W_{0}}{2}\right)^{1/(\beta-2)}.
\end{align}

\subsubsection{Scattering Length for a Regularized $W_{\beta}$ in 3D}\label{app:w_cut_calc}
Let $B_{\bx,R}$ denote the ball of radius $R$ around the point $\bx$. We define
\begin{align}\label{def:w_cut}
	W_{\beta}^R(\bx) = \frac{W_{0}}{R^{\beta}}\chi_{B_{0,R}}(\bx),
\end{align}
a regularized version of $W_{\beta}$. To compute the scattering length for this potential, we look for a solution $\vphi(\bx)$ to the equation
\begin{align}\label{eq:scattering_w_cut}
	\left(-\Delta + \frac{1}{2}W_{\beta}^R(\bx)\right)\vphi(\bx) = 0
\end{align}
with the asymptotics $\vphi(\bx) = 1-a_{\beta}^R |\bx|^{-1} + O(|\bx|^{-2})$, $a_{\beta}^R$ being the scattering length. Radial symmetry and the substitution $u(r) = r \vphi(r)$ shows that \eqref{eq:scattering_w_cut} is equivalent to the differential equation
\begin{align*}
	-u''(r) + \frac{1}{2}W_{\beta}^R(r)u(r) = 0, \quad u(0) = 0,
\end{align*}
$u(r)$ having asymptotics $u(r) = r - a_{\beta}^R + \mathcal{O}(r^{-1})$ for large $r$. After solving this simple boundary value problem one obtains
\begin{align*}
	u(r) = \left\{ \begin{array}{ll}
 				\frac{R^{\beta/2}}{\sqrt{W_{0}/2}}
				\frac{\sinh(\sqrt{W_{0}/2}R^{-\beta/2}\,r)}{\cosh(\sqrt{W_{0}/2}\,R^{1-{\beta/2}})}, 
				&0<r<R,\\
				r-a_{\beta}^R, &r \geq R,
			\end{array} \right.
\end{align*}
with
\begin{align}\label{eq:scattering_l_w_cut}
	a_{\beta}^R 
	= R-\frac{R^{\beta/2}}{\sqrt{W_{0}/2}}\tanh\left(\sqrt{W_{0}/2}\,R^{1-\beta/2}\right).
\end{align}

\subsubsection{Scattering Length for $W_{\beta}$ in 2D}
In this subsection we compute the two-dimensional scattering length for the potential
\begin{align*}
	W_\beta(\bx) = W_{0} |\bx|^{-\beta}, \quad \beta > 2. 
\end{align*}
Through scaling we see that if $\psi(\bx)$ is a solution to the equation
\begin{align}\label{eq:scattering_w_alpha_simple_2d}
	(-\Delta + |\bx|^{-\beta})\psi(\bx) = 0,
\end{align}
then $\vphi(\bx) = \psi((W_{0}/2)^{-1/(\beta-2)}\bx)$ solves
\begin{align}\label{eq:scattering_w_alpha_2d}
	\left(-\Delta + \frac{1}{2}\frac{W_{0}}{|\bx|^{\beta}}\right)\vphi(\bx) = 0
\end{align}
so we first look for solutions to \eqref{eq:scattering_w_alpha_simple_2d} 
with $\psi(\bx) = \ln(|\bx|/a) + O(|\bx|^{-1})$ as $|\bx| \to \infty$,
$a$ being the scattering length in two dimensions.

Radial symmetry and the substitution $u(r) = \sqrt{r}\psi(r)$ shows that 
\eqref{eq:scattering_w_alpha_simple_2d} is equivalent to the differential equation
\begin{align*}
	-u''(r) - \frac{u(r)}{4r^2}+ r^{-\beta}u(r) = 0, \quad u(0) = 0,
\end{align*}
$u(r)$ having asymptotics $u(r) = \sqrt{r}\ln(r/a) + \mathcal{O}(r^{1/2})$ for large $r$. We make the Ansatz
\begin{align*}
	u(r) = \sqrt{r}F(2r^{-(\beta-2)/2}/(\beta-2)),
\end{align*}
with $F$ a function to be determined. Let us for convenience define 
$t := 2r^{-(\beta-2)/2}/(\beta-2) = 2\nu r^{-1/(2\nu)}$, with 
$\nu := 1/(\beta-2) \in (0,1)$. Explicit computation shows that $F(t)$ 
satisfies the modified Bessel equation
\begin{align*}
	t^2 F''(t) + tF'(t) - t^2F(t) = 0,
\end{align*}
so $F(t) = c_1\mathcal{I}_{0}(t) + c_2K_{0}(t)$, where $\mathcal{I}_{0}, K_{0}$ are 
modified Bessel functions of order zero;
see \cite[Chapter 9.6 \& 9.7]{Abramowitz-Stegun:70}. 
The requirement $u(0)=0$ reduces to the condition
\begin{align*}
	\lim_{t \to \infty}(t/(2\nu))^{-\nu}F(t) = 0,
\end{align*}
hence $c_1 = 0$ and $F(t) = c_2K_{0}(t)$. For large $r$ one has
\begin{align*}
	\sqrt{r}K_{0}(2\nu r^{-1/(2\nu)}) 
	&= C\left(\sqrt{r}\ln\left(2\nu r^{-1/(2\nu)}\right)+\mathcal{O}(r^{1/2})\right)\\
	&= C\left(\sqrt{r}\ln\left(((2\nu)^{-2\nu} r)^{-1/(2\nu)}\right)+\mathcal{O}(r^{1/2})\right).
\end{align*}
From this we can conclude that the solution (with the right normalization) to 
\eqref{eq:scattering_w_alpha_simple_2d} has $a = (2\nu)^{2\nu}$,
so expanding $\vphi(\bx) = \psi((W_{0}/2)^{-1/(\beta-2)}\bx)$ shows that
$W_{\beta}$ has scattering length
\begin{align}\label{eq:scattering_l_homogeneous_2d}
	a_{\beta} = \left(\frac{2}{\beta-2}\right)^{2/(\beta-2)}\left(\frac{W_{0}}{2}\right)^{1/(\beta-2)}.
\end{align}

\subsubsection{Scattering Length for a Regularized $W_{\beta}$ in 2D}\label{app:w_cut_calc_2d}
Again, let $W_{\beta}^R(\bx) = W_{0}R^{-\beta}\chi_{B_{0,R}}(\bx)$
be a regularized version of $W_{\beta}$. We look for a solution $\vphi(\bx)$ 
to the equation
\begin{align}\label{eq:scattering_w_cut_2d}
	\left(-\Delta + \frac{1}{2}W_{\beta}^R(\bx)\right)\vphi(\bx) = 0
\end{align}
with the asymptotics $\vphi(\bx) = \ln(|\bx|/a_{\beta}^R) + \mathcal{O}(|\bx|^{-1})$, 
$a_{\beta}^R$ being the scattering length. By radial symmetry, \eqref{eq:scattering_w_cut_2d}
is equivalent to
\begin{align*}
	-\partial_r^2\vphi(r) - \frac{1}{r}\partial_r\vphi(r) + C\vphi(r) = 0
\end{align*}
for $r \in (0,R)$, where $C:=W_{0}R^{-\beta}$, and
\begin{align*}
	-\partial_r^2\vphi(r) - \frac{1}{r}\partial_r\vphi(r) = 0,
\end{align*}
for $r\geq R$. The first equation has the general solution
$\vphi(r) = c_1\mathcal{I}_0(\sqrt{C}r) + c_2K_0(\sqrt{C}r)$, however the condition 
$\vphi \in H^1(B_{0,R})$ forces $c_2 = 0$. As for the second equation we directly obtain
$\vphi(r) = \ln(r/a_{\beta}^R)$. 

This boundary value problem then has the solution
\begin{align*}
	\vphi(r) = \left\{ \begin{array}{ll}
 				\frac{R^{\beta/2-1}}{\sqrt{W_{0}/2}}
				\frac{\mathcal{I}_0\left(\sqrt{W_{0}/2}R^{-\beta/2}r\right)}
				{\mathcal{I}_1\left(\sqrt{W_{0}/2}R^{1-\beta/2}\right)}, 
				&0<r<R,\\
				\ln(r/a_{\beta}^R), &r \geq R,
			\end{array} \right.
\end{align*}
with scattering length
\begin{align}\label{eq:scattering_l_w_cut_2d}
	a_{\beta}^R = R\exp\left(-\frac{R^{\beta/2-1}}{\sqrt{W_{0}/2}}
	\frac{\mathcal{I}_0\left(\sqrt{W_{0}/2}R^{1-\beta/2}\right)}
	{\mathcal{I}_1\left(\sqrt{W_{0}/2}R^{1-\beta/2}\right)}\right).
\end{align}

\subsection{A Concave Lower Bound in the Hard-Disk Case}\label{app:hd_concave_calc}
We are looking for a constant $c > 0$ such that
\begin{align*}
	f(\gamma) = \frac{1}{c+\left(-\ln(2^{-1/2}\gamma)\right)_+}
\end{align*}
is concave in $\gamma$. Clearly, if $\gamma \geq \sqrt{2}$, then $f(\gamma) = c^{-1}$, which is concave.

We are thus left to study the case when $\gamma < \sqrt{2}$, and we choose to rewrite $f(\gamma)$ as
$f(\gamma) = -1/\ln(\alpha \gamma)$, with $\alpha = 2^{-1/2}e^{-c}$. Hence
\begin{align*}
	f''(\gamma) = -\frac{1}{\gamma^2\ln(\gamma\alpha)^2}\left(1+\frac{2}{\ln(\gamma\alpha)}\right),
\end{align*}
so that concavity and $\ln(\gamma\alpha)<0$ requires $e^2 < (\gamma\alpha)^{-1}$ for all $\gamma < \sqrt{2}$. We see that we can choose $\alpha = 2^{-1/2}e^{-2}$, hence $c = 2$.


\begin{thebibliography}{99}

\bibitem{Abramowitz-Stegun:70} 
M. Abramowitz, I. A. Stegun (eds.), \textit{Handbook of Mathematical Functions with Formulas, Graphs, and Mathematical Tables}, New York: Dover Publications, 9th-Printing (1970).

\bibitem{Adams-Fournier:03}
R. Adams, J. Fournier, \textit{Sobolev Spaces}, 2nd Ed., Academic Press (2003).

\bibitem{Anderson-et-al:95}
 M. H. Anderson, J. R. Ensher, M. R. Matthews, C. E. Wieman, E. A. Cornell, \textit{Observation of Bose-Einstein Condensation in a Dilute Atomic Vapor}, Science 269 (1995) 198-201.

\bibitem{Beliaev:58}
S. T. Beliaev, \textit{Energy spectrum of a non-ideal Bose gas}, Zh. Eksp. Teor. Fiz. 34 (1958), 433-446; Engl. Translation: Sov. Phys. JETP 7 (1958) 299-307.

\bibitem{Bloch-Dalibard-Zwerger:08}
I. Bloch, J. Dalibard, W. Zwerger, \textit{Many-body physics with ultracold gases}, Rev. Mod. Phys. 80 (2008) 885-964.

\bibitem{Brueckner-Sawada:57}
K. A. Brueckner, K. Sawada, \textit{Bose-Einstein Gas with Repulsive Interactions: General Theory}, Phys. Rev. 106 (1957) 1117-1127.

\bibitem{Brueckner-Sawada-2:57}
K. A. Brueckner, K. Sawada, \textit{Bose-Einstein Gas with Repulsive Interactions: Hard Spheres at High Density}, Phys. Rev. 106 (1957) 1128-1135.

\bibitem{DGPS:99}
F. Dalfovo, S. Giorgini, L. P. Pitaevskii, S. Stringari, \textit{Theory of Bose-Einstein condensation in trapped gases}, Rev. Mod.. Phys. 71 (1999) 463-512.

\bibitem{Davies-et-al:95}
K. B. Davis,  M.-O. Mewes, M. R. Andrews, N. J. van Druten, D. S. Durfee, D. M. Kurn, W. Ketterle, \textit{Bose-Einstein Condensation in a Gas of Sodium Atoms}, Phys. Rev. Lett. 75 (1995) 3969-3973.

\bibitem{Dyson:57}
F. J. Dyson, \textit{Ground-State Energy of a Hard-Sphere Gas}, Phys. Rev. 106 (1957) 20-26.

\bibitem{Dyson:68}
F. J. Dyson, \textit{Stability of Matter}, in Statistical Physics, Phase Transitions and Superfluidity, Brandeis University Summer Institute in Theoretical Physics 1966, pp. 179--239, (Gordon and Breach Publishers, New York, 1968).

\bibitem{Dyson-Lenard:67}
F. J. Dyson, A. Lenard, \textit{Stability of Matter. I}, J. Math. Phys. 8 (1967) 423-434.

\bibitem{Dyson-Lenard:68}
F. J. Dyson, A. Lenard, \textit{Stability of Matter. II}, J. Math. Phys. 9 (1968) 698-711.

\bibitem{Fefferman:83}
C. Fefferman, \textit{The Uncertainty Principle}, Bull. AMS 9 (1983) 129-206.

\bibitem{Frank-Seiringer:12}
R. L. Frank, R. Seiringer, \textit{Lieb-Thirring Inequality for a Model of Particles with Point Interactions}, J. Math. Phys. 53, 095201 (2012).

\bibitem{Girardeau-Arnowitt:59}
M. Girardeau, R. Arnowitt, \textit{Theory of Many-Boson Systems: Pair Theory}, Phys. Rev. 113 (1959) 755-761.

\bibitem{Giuliani-Seiringer:09}
A. Guiliani, R. Seiringer, \textit{The Ground State Energy of the Weakly Interacting Bose Gas at High Density}, J. Stat. Phys. 135 (2009) 915-934.

\bibitem{H-O2:77}
M. Hoffmann-Ostenhof, T. Hoffmann-Ostenhof, \textit{``Schr\"odinger inequalities" and asymptotic
behaviour of the electron density of atoms and molecules}, Phys. Rev. A 16 (1977) 1782-1785.

\bibitem{H-O2-Laptev-Tidblom:08}
M. Hoffmann-Ostenhof, T. Hoffmann-Ostenhof, A. Laptev, J. Tidblom, \textit{Many-particle Hardy Inequalities}, J. London Math. Soc. (2) 77 (2008) 99-114.

\bibitem{Huang-Yang:57}
K. Huang, C. N. Yang, \textit{Quantum-Mechanical Many-Body Problem with Hard-Sphere Interaction}, Phys. Rev. 105 (1957) 767-775.

\bibitem{Hugenholty-Pines:59}
N. Hugenholtz, D. Pines, \textit{Ground-State Energy and Excitation Spectrum of a System of Interacting Bosons}, Phys. Rev. 116 (1959) 489-506.

\bibitem{Lee-Yang:60}
T. D. Lee, C. N. Yang, \textit{Many-Body Problem in Quantum Statistical Mechanics. III. Zero-Temperature Limit for Dilute Hard Spheres}, Phys. Rev. 117 (1960) 12-21.

\bibitem{Lee-Huang-Yang:57}
T. D. Lee, K. Huang, C. N. Yang, \textit{Eigenvalues and Eigenfunctions of a Bose System of Hard Spheres and Its Low-Temperature Properties}, Phys. Rev. 106 (1957) 1135-1145.

\bibitem{Lenard:73}
A. Lenard, \textit{Lectures on the Coulomb Stability Problem}, in Statistical mechanics and mathematical problems, Battelle Rencontres, Seattle, Wash., 1971, Lect. Notes Phys., Vol. 20, pp. 114--135, 1973.

\bibitem{Lieb-Liniger:63}
E. H. Lieb, W. Liniger, \textit{Exact analysis of an interacting bose gas. I. The general
solution and the ground state}, Phys. Rev. 130 (1963) 1605-1616.

\bibitem{Lieb-Loss:01}
E. H. Lieb, M. Loss, \textit{Analysis}, 2nd edn., AMS, Providence (2001).

\bibitem{Lieb-Seiringer:10}
E. H. Lieb, R. Seiringer, \textit{The stability of matter in quantum mechanics}, Cambridge University Press, Cambridge, (2010).

\bibitem{Lieb-Seiringer-Yngvason:03}
E. H. Lieb, R. Seiringer, J. Yngvason, \textit{One-Dimensional Bosons in Three-Dimensional Traps}, Phys. Rev. Lett. 91 (2003) 150401.

\bibitem{Lieb-Seiringer-Yngvason:04}
E. H. Lieb, R. Seiringer, J. Yngvason, \textit{One-Dimensional Behavior of Dilute, Trapped Bose Gases}, Commun. Math. Phys. 244 (2004) 347-393.

\bibitem{Lieb-Seiringer-Solovej-Yngvason:05}
E. H. Lieb, R. Seiringer, J. P. Solovej, J. Yngvason, \textit{The Mathematics of the Bose Gas and its Condensation}, (Series: Oberwolfach Seminars, Vol. 34, Birkh\"auser Verlag, 2005)

\bibitem{Lieb-Thirring:75}
E. H. Lieb, W. Thirring, \textit{Bound for the Kinetic Energy of Fermions which Proves the Stability of Matter}, Phys. Rev. Lett. 35 (1975) 687-689. Errata \emph{ibid.}, 1116.

\bibitem{Lieb-Thirring:76}
E. H. Lieb, W. Thirring, \textit{Inequalities for the Moments of the Eigenvalues of the Schr\"odinger Hamiltonian and Their Relation to Sobolev Inequalities}, in Studies in Mathematical Physics, pp. 269-303, Princeton University Press, 1976.

\bibitem{Lieb-Yngvason:98}
E. H. Lieb, J. Yngvason, \textit{Ground State Energy of the Low Density Bose Gas}, Phys. Rev. Lett. 80 (1998) 2504-2507.

\bibitem{Lieb-Yngvason:00}
E. H. Lieb, J. Yngvason, \textit{The Ground State Energy of a Dilute Bose Gas}, in: Differential Equations and Mathematical Physics, University of Alabama, Birmingham, 1999, R. Weikard and G. Weinstein, eds., pp. 271-282, Amer. Math. Soc./Internat. Press (2000).

\bibitem{Lieb-Yngvason:01}
E. H. Lieb, J. Yngvason, \textit{The Ground State Energy of a Dilute Two-Dimensional Bose Gas}, J. Stat. Phys. 103 (2001) 509-526.

\bibitem{Lundholm-Nam-Portmann}
D. Lundholm, P. T. Nam, F. Portmann, \textit{Fractional Hardy-Lieb-Thirring and related inequalities for interacting systems}, {\tt arXiv:1501.04570}

\bibitem{Lundholm-Solovej:anyon}
D. Lundholm, J. P. Solovej, \textit{Hardy and Lieb-Thirring inequalities for anyons}, Commun. Math. Phys. 322 (2013) 883-908.

\bibitem{Lundholm-Solovej:extended}
D. Lundholm, J. P. Solovej, \textit{Local exclusion and Lieb-Thirring inequalities for intermediate and fractional statistics}, Ann. Henri Poincar\'e 15 (2014) 1061-1107.

\bibitem{Lundholm-Solovej:exclusion}
D. Lundholm, J. P. Solovej, \textit{Local exclusion principle for identical particles obeying intermediate and fractional statistics}, Phys. Rev. A 88 (2013) 062106.

\bibitem{Pitaevskii-Stringari:03}
L. Pitaevskii, S. Stringari, \textit{Bose-Einstein Condensation}, Oxford Science Publications, Oxford (2003).

\bibitem{Schnee-Yngvason:06}
K. Schnee, J. Yngvason, \textit{Bosons in Disc-Shaped Traps: From 3D to 2D}, Commun. Math. Phys. 269 (2006) 659-691.

\bibitem{Wu:59}
T. T. Wu, \textit{Ground State of a Bose System of Hard Spheres}, Phys. Rev. 115 (1959) 1390-1404.

\bibitem{Yau-Yin:09}
H.-T. Yau, J. Yin, \textit{The Second Order Upper Bound for the Ground Energy of a Bose Gas}, J. Stat. Phys. 136 (2009) 453-503.

\end{thebibliography}
\end{document}